\documentclass[12pt]{article}
\usepackage{amsmath,amsthm,amsfonts}
\usepackage{graphicx,psfrag,epsf}
\usepackage{enumerate}
\usepackage{url}
\usepackage{subfig}
\usepackage{multirow}
\usepackage{natbib}
\usepackage{amssymb}
\usepackage{bbm}
\usepackage{placeins}
\usepackage{float}
\usepackage{afterpage}
\usepackage{array}
\usepackage{multirow}
\usepackage{setspace}
\usepackage{epsfig}
\usepackage{booktabs}
\usepackage{comment}
\usepackage{xcolor}
\usepackage{float}
\usepackage[title]{appendix}

\usepackage{xr}
\makeatletter
\newcommand*{\addFileDependency}[1]{
  \typeout{(#1)}
  \@addtofilelist{#1}
  \IfFileExists{#1}{}{\typeout{No file #1.}}
}
\makeatother

\usepackage[colorinlistoftodos]{todonotes}
\usepackage{algorithm2e}
	\SetKwComment{Comment}{\#\#\#}{}
	\DontPrintSemicolon
	\SetAlgoLined
	\RestyleAlgo{boxruled}
\usepackage{titlesec}

\DeclareMathOperator*{\argmax}{arg\,max}
\DeclareMathOperator*{\argmin}{arg\,min}

\newcommand{\Cov}{\textrm{Cov}}
\newcommand{\E}{\textrm{E}}
\newcommand{\pred}{\textrm{pred}}

\theoremstyle{plain}
\newtheorem{theorem}{Theorem}

\theoremstyle{definition}

\newtheorem{proposition}{Proposition}[section]

\theoremstyle{remark}

\newcommand{\mbf}{\mathbf}

\newcommand{\blind}{1}

\begin{document}

\date{}

\def\spacingset#1{\renewcommand{\baselinestretch}%
{#1}\small\normalsize} \spacingset{1}


\if1\blind
{
  \title{\bf $\mbf{P^3LS}$: Point Process Partial Least Squares}
  \author{Jamshid Namdari\thanks{
    Corresponding author Jamshid Namdari, Department of Biostatistics \& Bioinformatics, Emory University, Atlanta, GA 30322 (e-mail:jamshid.namdari@emory.edu). This work is supported by National Institutes of Health grants R01GM140476, R01HL159213 and R01MH125816.} \\
    Department of Biostatistics \& Bioinformatics, Emory University\\
    Robert T. Krafty \\
    Department of Biostatistics \& Bioinformatics, Emory University\\
    and\\
    Amita Manatunga \\  
    Department of Biostatistics \& Bioinformatics, Emory University}
  \maketitle
} \fi

\if0\blind
{
  \bigskip
  \bigskip
  \bigskip
  \begin{center}
    {\LARGE\bf $\mbf{P^3LS}$: Point Process Partial Least Squares}
\end{center}
  \medskip
} \fi

\bigskip
\begin{abstract}

Many studies collect data that can be considered as a realization of a point process. 
Partial least squares (PLS) is a popular analytic approach that combines features from linear modeling as well as dimension reduction to provide parsimonious prediction and classification. However, existing PLS methodologies do not include the analysis of point process predictors. In this article, we introduce point process PLS ($P^3LS$) for analyzing latent time-varying intensity functions from collections of inhomogeneous point processes. We develop a novel estimation procedure for $P^3LS$ that utilizes the properties of log-Gaussian Cox processes and  examine its empirical properties via simulation studies. In addition, we establish  the predictive consistency for the $P^3LS$. Finally, we apply the proposed method to a renal radionuclide imaging study to predict kidney obstruction from patient renograms.
 The method yields a clinically meaningful and interpretable model, demonstrating its potential to support data-driven clinical decision-making in renal studies.
\end{abstract}

\noindent%
{\it Keywords:} Partial least squares; Point process; Log Gaussian Cox process; Functional linear model; Adaptive dimension reduction.

\vfill

\newpage
\spacingset{1.9} 

\section{Introduction} \label{sec:Introduction}
Partial Least Squares (PLS), originally proposed by \cite{Wold196},
has emerged as a promising strategy for predicting a response in terms of a covariate.  Under linear regression, the classical PLS approach targets on maximizing predictive power while achieving dimension reduction in a supervised manner to
extract a set of orthogonal latent factors from predictors.
The method has found particular popularity in chemometrics \citep{mehmood2016diversity} as well as in many other scientific fields including bioinformatics \citep{boulesteix2007partial,datta2018exploring}, food science \citep{castro2020partial}, pharmacology \citep{luco1999prediction},  and psychology \citep{ronkko2015adoption}.
To broaden its applicability, several extensions have been proposed to accommodate more complex data structures, such as multilevel and functional data.
For more on the developments and applications in scientific fields refer to \cite{rosipal2005overview}, \cite{abdi2010partial}, and \cite{krishnan2011partial}. 
Recent technological advancements have enabled the generation of increasingly complex data structures, such as medical images, which can significantly enhance disease diagnosis and improve outcomes in scientific research. For example, medical images developed using radioactive agents often produce data that can be viewed as realizations of point processes. Incorporating this information from the data generation process is likely to improve the interpretation and prediction of clinical outcomes of interest. To the best of our knowledge, PLS's extension to incorporate temporal point processes has not been explored.

Point process data is a realization of a random set of points in a specified space such as time or a plane. In particular, temporal point process data can be viewed as the recorded times at which events of interest occur. Common illustrative examples include arrival times of a patient to an emergency room and patient engagement with mobile phone-based telehealth applications.
Our motivating study, described in Section \ref{sec:Data_Analysis}, involves a radionuclide imaging experiment in which subject-specific temporal point process data, photon detection times are collected to aid radiologists in interpreting potential kidney obstruction.
In this imaging protocol, data acquisition begins immediately following the intravenous administration of the gamma-emitting radiotracer Technetium-99m mercaptoacetyltriglycine (99mTc-MAG3). This tracer is extracted from the bloodstream by the kidneys and then passes through the ureters to the bladder. As MAG3 travels through the kidneys, emitted photons are detected by a gamma camera, producing a sequence of arrival times of photons over the 24-minute period.
For each subject, the resulting data comprise a sequence of photon arrival times, representing a realization of a temporal point process that captures the tracer dynamics through the renal system.
Following image acquisition, expert radiologists interpret the scans and assign a continuous score to each kidney, reflecting the degree of suspected obstruction. These scores range from -1 to 1, where values near 1 indicate high confidence in the presence of obstruction, while values near -1 suggest high confidence in the absence of obstruction.
Accurate interpretation of renogram data demands substantial expertise in renal physiology and MAG3. However, diagnostic performance can vary considerably, especially in settings where most evaluations are performed by radiologists with limited nuclear medicine training and without expert consultation.  
Reported inter-reader discrepancies range from $9\%$ to $72\%$ \citep{taylor2008diagnostic}. A challenge is the lack of a definitive gold standard for diagnosing obstruction with radionuclide imaging, which further contributes to interpretive variability. In this context, developing an interpretable predictive model that supports clinical decision-making could be highly valuable.
Such a model aims to support clinical decision-making, particularly in settings where radiologists may have limited experience, by serving as a quantitative, data-driven second opinion to assist in the interpretation of  kidney obstruction.

Towards the goal of fitting an interpretable predictive model for predicting a response of interest (degree of obstruction) based on a realization of a point process, we formulate a sensible model, via linear functional model (\ref{eq:Regression_model}) in Section \ref{sec:Methodology}, which relates the response based on the unobserved log-intensity functions of the underlying point process. This linear model could be viewed as an extension of the functional linear model of \cite{Delaigle_Hall_2012} to point process data.
To establish a functional linear relationship between the predictors, which are log-intensity functions of the underlying point processes, and the response,
our approach  selects basis functions adaptively to maximize the predictive power of the linear model.  In contrast to the common approach of representing functions in a pre-selected set of basis functions \citep{Ramsay_Silverman_2005}, our result yields a more interpretable model that is parsimoniously optimal.
Two key challenges in fitting such models, where latent log-intensity functions are not directly observed, are (1) the estimation of the covariance function of the generating log-intensity process, which is necessary for the estimation of the parsimonious basis, and (2) the prediction of individual log-intensities  within the parsimonious basis, which are necessary for estimating the coefficient quantifying the association between log-intensities and the outcome. We develop a novel procedure that utilizes the properties of the log Gaussian Cox process \citep{moller1998log} to construct  efficient estimators of these functional quantities, and represents the first extension of PLS to incorporate log Gaussian Cox point process predictors.

This paper is organized as follows. In Section \ref{sec:Methodology} we describe our approach to point process partial least squares and present its asymptotic properties. To illustrate the performance of the proposed methodology and compare it to alternative approaches, we present a simulation study in Section \ref{sec:Simulation}.  Section \ref{sec:Data_Analysis} contains a detailed description of the renal study that motivated our methodological development and the results of the application the $P^3LS$ to the study data. 
We conclude this paper, in Section \ref{sec:Discussion}, by a discussion of the limitations and strengths of the proposed method and possible directions of further work related to point process partial least squares.

\section{Methodology} \label{sec:Methodology}
The data considered here are $n$ independent pairs $(\Phi_1,Y_1),\dots,(\Phi_n,Y_n)$, where $Y_1,\dots,Y_n$ are the outcomes (scalar) and $\Phi_1,\dots,\Phi_n$ are realizations of the point process, described in the next paragraph. We seek to build a predictive model that can predict $Y_i$'s using features in the $\Phi_i$'s that parsimoniously represent the dynamics of the point processes. This is achieved by modeling $Y_i$'s through a linear model incorporating latent intensity functions, $\lambda_i(.)$, that govern the dynamics of the process $\Phi_i$, and applying the partial least squares regression, as described in this section after providing a concise introduction to the point process concepts. The partial least squares method along with the estimation procedure proposed in Section \ref{sec:Estimation} will be referred to as $P^3LS$ in this article. In what follows, we describe our model building and model fitting procedure.  
    
Consider an event of interest, such as detection of a photon by a gamma camera, since the beginning of scanning, or time of engaging with a mobile application, since registration in a study. For each subject  $i=1,\dots,n$, the set of times of occurrence of the event of interest $\Phi_i=\{S_{i1},S_{i2},\dots\}\subset\mathbb{R}$, can be viewed as a random set that is referred to as a point process. Let $\mathcal{I}$ be the time interval over which the point process is observed. Note that, $\mathcal{I}$ can be any compact subset of the real line, including a simple continuous interval, or more complicated structures. The point process $\Phi_i$ can be described in terms of the total number of occurrences of the event of interest in any interval $B\subseteq\mathcal{I}$, such as $B=(t_1,t_2], t_1,t_2\in\mathcal{I}$, which we denote by $N_i(B)$. This allows us to study the statistical properties of the point process by modeling the probability distribution of $N_i(.)$ over any measurable subset of the real line. We assume that $\Phi_i$'s follow the construction of Poisson processes, that is $\mathcal{I}$ can be divided into small subintervals of length $\Delta$ where in each interval {e.g. $(t,t+\Delta], t\in\mathcal{I},$} the event of interest occurs with a positive probability, but the probability of occurrence of more than one event is negligible. In addition, we would like to allow the probability of occurrence of the event  of interest to depend on time, in other words the instantaneous probability of occurrence of an event at each time point, $t$, can be a function $\lambda_i(t)$ of time.
Finally, we assume that the occurrence of an event at time $t$ does not have excitatory or inhibitory effect on the occurrence of an event at a later time. In other words, we can assume that for two non-overlapping intervals $B_1$ and $B_2$, $N_i(B_1)$ and $N_i(B_2)$ are independent. Note that $\lambda_i(t)$ can be viewed as $\lim_{\Delta\to 0}E[N_i(t,t+\Delta)]/\Delta$, which indicates the rate at which the event occurs at time $t$. Lastly, to account for variability among experimental units, we assume that $\lambda_i(t)$ is a random function such that $\log\left[\lambda_i(t)\right]$ follows a Gaussian process. The point process described above is called log-Gaussian Cox process.

To be more precise, consider the point processes $\Phi_1,\dots,\Phi_n$. For each $\Phi_i, i=1,\dots,n$, let $N_i(B) = \textrm{card}\{\Phi_i\cap B\}$ be the number of events of $\Phi_i$ in a Borel set $B\subset\mathbb{R}$ and define the intensity measure of $\Phi_i$ to be $\Lambda_i(B)=\E[N_i(B)]$ with the intensity function $\lambda_i(x)$, i.e. $\Lambda_i(B)=\int_B\lambda_i(x)dx$. $\Phi_i$ is called a Poisson process on $\mathcal{I}\subset\mathbb{R}$ with intensity measure $\Lambda_i$ if for any $B\subset \mathcal{I}$ 
\begin{itemize}
    \item $N_i(B)$ is Poisson distributed with mean $\Lambda_i(B)$,
    \item conditional on $N_i(B)$, the points in $\Phi_i\cap B$ are iid with density proportional to $\lambda_i(x),\; x\in B$.
\end{itemize}
Moreover, $\Phi_i$ is called a Cox process driven by a non-negative process $\lambda_i$ if, conditional on $\lambda_i$, $\Phi_i$ is a Poisson process with intensity function $\lambda_i$. In this paper we consider the log Gaussian Cox Poisson process, introduced by \cite{moller1998log}, where $\log\left[\lambda_i(t)\right]= \log\left[\lambda_0(t)\right] + \Xi_i(t)$, and $\Xi_i$`s are zero-mean independent Gaussian processes with a common covariance function $K(s,t) = \textrm{Cov}\left[\Xi_i(s),\Xi_i(t)\right]$. Suppose $\Phi_1,\dots,\Phi_n$ are $n$ realizations of a log-Gaussian Cox process driven by $\lambda_1,\dots,\lambda_n$, respectively. In this article, we denote the log-intensities by $X_i:=\log(\lambda_i)$.

To introduce the predictive model, consider the independent pairs $(X_1,Y_1),\dots,(X_n,Y_n)$, where $X_i,i=1,\dots,n$ are the log-intensity functions defined on the nondegenerate, compact interval $\mathcal{I}$ and satisfying $\int_\mathcal{I}\E(X_i^2) < \infty$, and $Y_1,\dots,Y_n$ are scalar random variables generated by the following linear model
\begin{equation} \label{eq:Regression_model}
    Y_i = a + \int_{\mathcal{I}}b(t)X_i(t)\;dt + \epsilon_i, \quad i=1,\dots,n.
\end{equation}
Here $a$ is a scalar parameter, $\epsilon_i, i=1,\dots,n$ are iid random variables with finite second moment such that $\E(\epsilon_i|X_i)=0$, and $b$, a function valued parameter, is a square integrable function on $\mathcal{I}$. 
The class of square integrable functions on $\mathcal{I}$ we consider here, denoted by $\mathcal{C}(\mathcal{I})$, is equipped with the inner product and the norm defined as $\langle u,v \rangle := \int_{\mathcal{I}}u(s)K(s,t)v(t)dsdt$ and $\|u\| := \sqrt{\langle u,u\rangle}$, where $K(s,t) = \Cov\left[X_i(s),X_i(t)\right]$, for $u,v\in\mathcal{C}(\mathcal{I})$. Note that the condition $\E(\epsilon_i|X_i)=0$ implies $a=\E(Y_i)-\int_\mathcal{I}b(t)\E\left[X_i(t)\right]\;dt$, so $Y_i=\E(Y_i) + \int_\mathcal{I}b(t)\left\{X_i(t)-\E\left[X_i(t)\right] \right\}\;dt + \epsilon_i$.

To estimate the coefficient function in (\ref{eq:Regression_model}), typically one expands $X_i$'s and $b$ in a system of orthonormal basis functions, $\{\psi_1,\psi_2,\dots\}$, and estimate $b$ by finding optimal coefficients in the truncated expansion, $b_p$ of $b$, where 
\begin{equation} \label{eq:coefficient_function_expansion}
        b_p = \sum_{j=1}^{p}\beta_j\psi_j, 
\end{equation}
and $\beta_1,\dots,\beta_p$ are coefficients corresponding to basis functions. Note that, by approximating $b$ with $b_p$, the truncated form of the linear functional $a+\int_\mathcal{I}b(t)X_i(t)\;dt$ can be written as,  say $g_p(X_i)$, where
\begin{equation} \label{eq:g_p}
        g_p(X_i) := \E(Y_i) + \sum_{j=1}^{p}\beta_j\int_\mathcal{I}\left\{X_i(t)-\E\left[X_i(t)\right]\right\}\psi_j(t)\; dt.
\end{equation}
Then, we can approximate $Y_i$ by $g_p(X_i)+\epsilon_i$ and determine $\beta_1,\dots,\beta_p$ through the least squares method, i.e. by minimizing
\begin{equation} \label{eq:beta_hat}
    \beta_1,\dots,\beta_p = \argmin_{w_1,\dots,w_p}\frac{1}{n}\sum_{i=1}^{n}\left\{Y_i^c-\sum_{j=1}^{p}w_j\int_{\mathcal{I}}X_i^c(t)\psi_j(t)\;dt\right\}^2,
\end{equation}
where $Y_i^c = Y_i-\bar{Y}$, $X_i^c(t) = X_i(t)-\bar{X}(t)$, and $\bar{X}(t) = \sum_{j=1}^{n}X_j(t)/n$.

One adaptive procedure for selecting the basis functions that captures both the covariance structure of $X_i$'s as well as the linear relationship between $Y_i$ and $X_i$ is through the Partial Least Squares (PLS) regression. \cite{Delaigle_Hall_2012} proposed a functional partial least squares procedure for constructing the basis functions $\psi_1,\dots,\psi_p$ in a sequential manner, such that for $p=1$, $\psi_1$ is determined so that $\|\psi_1\|=1$ and $\Cov\left\{Y_i-E(Y_i),\right.$ $\left.\int_\mathcal{I}\left[X_i(t)-E(X_i(t))\right]\psi_1(t)\;dt\right\}$ is maximized. Note that, when $p=1$, $b_1 = \beta_1\psi_1$ and the linear model (\ref{eq:Regression_model}) reduces to simple linear regression model $$Y_i = \E(Y_i) + \beta_1\int_\mathcal{I}\left\{X_i(t)-\E\left[X_i(t)\right] \right\}\psi_1(t)\;dt + \epsilon_i,$$ where $\beta_1$ can be determined by the least squares method. Next, the second PLS basis function, $\psi_2$, is determined such that it is orthogonal to $\psi_1$ and the covariance between the deflated response (response after removing the linear effect of $\psi_1$ on the $Y_i$) and the projected data onto $\psi_2$ is maximized, i.e. $\psi_2 = \argmax_{\psi}\Cov\{Y_i-g_1(X_i), \int_\mathcal{I} \left(X_i(t)-\E\left[X_i(t)\right]\right)\psi(t)\}$ such that $\langle \psi_1,\psi_2 \rangle=0$ and $\|\psi_2\|=1$, where $g_1(X_i) = \E(Y_i) + \beta_1\int_\mathcal{I}\left\{X_i(t)-\E\left[X_i(t)\right]\right\}\psi_1(t)\; dt$. Sequentially, in the same manner, the $p$-th PLS basis function is constructed such that the covariance functional 
\begin{equation} \label{eq:Cov_functional}
    f_p(\psi_p) = \Cov\left\{Y_i-g_{p-1}(X_i),\int_\mathcal{I}X_i(t)\psi_p(t)\;dt\right\},
\end{equation}
is maximized subject to $\|\psi_p\|=1$ and $\langle \psi_j,\psi_p \rangle=0$ for $1\leq j\leq p-1$, where $g_p$ and $b_p$ are defined in equations (\ref{eq:g_p}) and (\ref{eq:coefficient_function_expansion}), representing the truncated expansions of the linear functional $a+\int_\mathcal{I}b(t)X_i(t)\;dt$ and the coefficient function $b$, with respect to $\psi_1,\dots,\psi_p$, respectively. For each $p\in\mathbb{N}$, $\beta_1,\dots,\beta_p$ are obtained by minimizing the mean squared error of prediction as in (\ref{eq:beta_hat}).
    
An interesting property of the PLS basis functions is that for each $p\geq 1$, the linear representation of any function in $\psi_1,\dots,\psi_p$ is equivalent to representing it as a linear combination of $K(b),\dots,K^p(b)$, where 
\begin{align} \label{eq:K(b)}
\begin{split}
    K(b)(t) &= \int_{\mathcal{I}}b(s)K(s,t)\;ds,  \\
K^j(b)(t) &= \int_{\mathcal{I}}K^{j-1}(s)K(s,t)\;ds, \quad j>1.
\end{split}
\end{align}
See \cite{Delaigle_Hall_2012} for more details. This motivates considering basis functions $\psi_1,\psi_2,\dots$ that are obtained by applying the modified Gram-Schmidt orthonormalization procedure (outlined in the Supplementary Materials) to $K(b),K^2(b),\dots$.
We adopt the above procedure for constructing PLS basis functions for the $P^3LS$ procedure in this paper.

In the context of point process data, we note that log-intensity functions, $X_1,\dots,X_n$, are not observable. That is, we need to estimate $X_j,j=1,\dots,n$, in addition to the covariance function, $K(s,t)$, using the realizations of the point process that we describe in the next section.

\subsection{Estimation Procedure} \label{sec:Estimation}
\subsubsection{Estimation of the Covariance Functions} \label{subsec:Cov_est_PP}
We proceed with estimation of the covariance function of the log-intensities by using their relation to the second order intensities, denoted and defined as $\rho_{i,j}^{(2)}(s,t) := \E[\lambda_i(s)\lambda_j(t)]$ for $i,j=1,\dots,n$, accompanied with an application of  Campbell's Theorem \citep{Daley_Vere-Jones_2003}.  

Using the moment generating function of the normal distribution 
\begin{equation}
    \E[\lambda_i(s)\lambda_i(t)] = \E[\lambda_i(s)]\E[\lambda_i(t)]\exp\left\{K(s,t)\right\}, \quad \mbox{ for } i=1,\dots,n.
\end{equation}
Thus,
\begin{equation}\label{eq:c(s,t)}
    K(s,t) = \log\frac{\E[\lambda_i(s)\lambda_i(t)]}{\E[\lambda_i(s)]\E[\lambda_i(t)]}, \quad \mbox{ for } i=1,\dots,n.
\end{equation}
In addition, since the $n$ random intensity functions are independent, $\E[\lambda_i(s)\lambda_j(t)] = $ \\ $\E[\lambda_i(s)]\E[\lambda_j(t)]$ for all $i\neq j$ and $\E[\lambda_i(t)]=\E[\lambda_j(t)]$ for all $i,j=1,\dots,n$. Thus, we can rewrite (\ref{eq:c(s,t)}) as
\begin{equation} \label{eq: c(s,t)_rho_relation}
    K(s,t) = \log\frac{\rho_{i,i}^{(2)}(s,t)}{\rho_{i,j}^{(2)}(s,t)}, \quad \mbox{ for } i,j=1,\dots,n \mbox{ and } i\neq j.    
\end{equation}
To estimate the second order intensity function $\rho_{i,j}^{(2)}(s,t)$, we invoke to the Campbell`s Theorem, which states that, for any measurable function, $f(u,v)$
\begin{equation} \label{eq:Campbell}
    \E\left[ \sum_{u\in\Phi_i}^{u\neq v}\sum_{v\in\Phi_j} f(u,v) \right] = \int\int f(u,v)\rho_{i,j}^{(2)}(u,v)\;du\;dv,
\end{equation}
where the expectation is over the point processes $\Phi_i$ and $\Phi_j$. \cite{xu2020semi} have shown that a consistent estimate of $\rho_{i,j}^{(2)}(s,t)$ can be obtained by selecting $f(u,v) = \kappa_h(s-u)\kappa_h(t-v)/[a(s;h)a(t;h)]$, where $\kappa(.)$ is a kernel function, $\kappa_h(u)=\kappa(u/h)/h$, and $a(s;h)=\int\kappa_h(s-x)dx$ is an edge correction term. This enables us to estimate $K(s,t)$, denoted by $\hat{K}(s,t)$, by the plug in estimator where  the numerator, $E[\lambda_i(s)\lambda_i(t)]$, and the denominator, $\E[\lambda_i(s)]\E[\lambda_i(t)]$, of (\ref{eq:c(s,t)}) are estimated by $\widehat{\E[\lambda_i(s)\lambda_i(t)]}$ and $\widehat{\E[\lambda_i(s)]\E[\lambda_i(t)]}$, respectively, given by
\begin{align}
  \widehat{\E[\lambda_i(s)\lambda_i(t)]} &= \frac{1}{n}\sum_{i=1}^{n}\sum_{x\in\Phi_i}^{x\neq y}\sum_{y\in\Phi_i} \frac{\kappa_h(s-x)\kappa_h(t-y)}{a(s;h)a(t;h)}, \\
  \widehat{\E[\lambda_i(s)]\E[\lambda_i(t)]} &= \frac{1}{n(n-1)}\underset{\begin{subarray}{l}
  i,j=1,\dots,n \\ 
  \quad i\neq j
  \end{subarray}}{\sum\sum} \; \sum_{x\in\Phi_i}^{x\neq y}\sum_{y\in\Phi_j} \frac{\kappa_h(s-x)\kappa_h(t-y)}{a(s;h)a(t;h)},
\end{align}
to obtain
\begin{equation} \label{eq:K_hat}
    \hat{K}(s,t) = \log\frac{\widehat{E[\lambda_i(s)\lambda_i(t)]}}{\widehat{\E[\lambda_i(s)]\E[\lambda_i(t)]}}.
\end{equation}

\subsubsection{Estimation of the Intensity Functions}\label{subsubsec:est_intensity}
Given the eigenfunctions $\phi_1,\phi_2,\dots$ of the covariance function $K(.,.)$, the log-intensities can be expanded as 
\begin{equation} \label{eq:log_intensity_expansion}
    \log\left[\lambda_i(t)\right] = \sum_\ell \xi_{i\ell}\phi_\ell(t).
\end{equation}
This motivates considering the following method for estimating the scores $\xi_{i\ell}, \ell=1,2,\dots$. We first partition the time interval, $\mathcal{I}$, into bins $B_1,\dots,B_b$ and denote the midpoint of each bin by $\bar{t}_1,\dots,\bar{t}_b$. Let $N_{i\ell},\ell=1,\dots,b$ be the number of events in the point process $\Phi_i$ falling in $B_\ell$. Note that, $N_{i\ell}\approx \textrm{Poisson}\left[\lambda_i(\bar{t}_\ell)\left|B_\ell\right|\right]$. Thus, $\E[N_{i\ell}/|B_\ell|]\approx\lambda_i(\bar{t}_\ell)$. Therefore, we can  consider the following log-linear model
\begin{equation}\label{eq:log-linear_model}
    \log\E[N_{i\ell}/|B_\ell|]) = \sum_\ell \xi_{i\ell}\phi_\ell(\bar{t}_\ell).
\end{equation}
Note that estimates of the eigenfunctions of the covariance function $K(.,.)$ can be obtained from those of $\hat{K}(.,.)$, defined in (\ref{eq:K_hat}). Suppose $\hat{K}(.,.)$ is evaluated on the grid $\mathcal{G} = \{t_1,\dots,t_T\}\star \{t_1,\dots,t_T\}$ and let $v_1,v_2,\dots$ be the corresponding eigenvectors of $\hat{K}$. Then, an estimate of the eigenfunctions $\phi_1,\phi_2,\dots$ evaluated on $\{t_1,\dots,t_T\}$ is $\hat{\phi}_\ell = v_\ell/\Delta$, where $\Delta = (t_T-t_1)/T$. This enables us to estimate the log-intensities by truncating (\ref{eq:log_intensity_expansion}) to the first $q$ terms, i.e. $X_i^{(q)} = \sum_{\ell=1}^{q} \xi_{i\ell}\phi_\ell(t)$, and plugging in $\hat{\phi_\ell}$ for $\phi_\ell$ and $\hat{\xi}_{i\ell}$, obtained through the log-linear model (\ref{eq:log-linear_model}), for $\xi_{i\ell}$, for $\ell=1,\dots,q$. The final estimate for $X_i$ is 
\begin{equation} \label{eq:X_hat}
    \hat{X}_i^{(q)} = \sum_{\ell=1}^{q} \hat{\xi}_{i\ell}\hat{\phi}_\ell(t).
\end{equation}

\subsubsection{Estimation of the Coefficient Function $b$}
Recall that the coefficient function $b$ can be estimated by truncation to the first $p$ terms of the expansion of $b$ with respect to the basis functions $\psi_1,\psi_2,\dots$, i.e. $b_p = \sum_{j=1}^{p}\beta_j\psi_j$. In addition, as described in Section \ref{sec:Methodology} following the equation (\ref{eq:K(b)}), $\psi_j$'s are obtained by applying the modified Gram-Schmidt algorithm to $K(b),K^2(b),\dots$. Here, we utilize the estimates obtained for $X_i$'s in (\ref{eq:X_hat}) and $\hat{K}(s,t)$ in (\ref{eq:K_hat}) to estimate $\psi_j$'s and $\beta_j$'s. To this end, first we estimate $K^{j}(b)(t)$ by $\hat{K}^{j}(b)(t),\; j\geq 1$ through
\begin{align*}
\hat{K}(b)(t) &= \frac{1}{n}\sum_{i=1}^{n}\left[\hat{X}_i^{(q)}(t)-\bar{\hat{X}}^{(q)}(t)\right]\left(Y_i-\bar{Y}\right) \\
\hat{K}^2(b)(t) &= \int_{\mathcal{I}}\hat{K}(b)(s)\hat{K}(s,t)ds \\
\hat{K}^{j+1}(b)(t) &= \int_{\mathcal{I}}\hat{K}^j(b)(s)\hat{K}(s,t)ds, \; j=1,2,\dots, 
\end{align*}
where $\bar{\hat{X}}^{(q)}(t) = \sum_{j=1}^{n}\hat{X}_i^{(q)}(t)/n$ and $\hat{K}(s,t)$ is estimated by (\ref{eq:K_hat}). Then we obtain the orthonormal basis $\psi_j,j=1,\dots$. Next we estimate $\beta_j$`s by solving (\ref{eq:beta_hat}). Denote the estimates obtained for $\beta_1,\dots,\beta_p$ by $\hat{\beta_1},\dots,\hat{\beta_p}$, then the final estimate for $b_p$ is 
\begin{equation}
    \hat{b}_p = \sum_{j=1}^{p}\hat{\beta}_j\hat{\psi}_j.
\end{equation}

The estimation procedure developed  depends on a number of  parameters. The parameters involved are $h$, the smoothing parameters, $q$, the number of terms used in (\ref{eq:X_hat}), $|B_\ell|$, size of the partitions of $\mathcal{I}$ as discussed in section \ref{subsubsec:est_intensity}, and $p$, the number of PLS basis functions used in estimation of the coefficient function $b$. To select $h$, one can consider cross validation where the optimal choice of the parameter minimizes the average mean square prediction error. Our simulation studies do not indicate that the mean square prediction error is sensitive to the choice of $h$. Next, $q$ can be selected large enough such that the variability explained by $\phi_1,\dots,\phi_q$ is above a specified threshold such as 90\% of the total variability. Theoretical considerations in Section \ref{subsec:Asymptotic_P3ls} suggest considering a partitioning of $\mathcal I$ such that the partition size is not larger than $|\mathcal{I}|n^{-1/2}$. Finally, an information criterion can be used to select $p$, which is supported by the linearity of the model along with an extra modeling assumption that the error process follows a Gaussian distribution.

Lastly, sometimes in practice, the set of event times (e.g. detection times in our motivating study) is not recorded. Instead, total counts within subintervals are available for analysis. More precisely, for the point process $\Phi_i$, one might only observe $N_i(B_j), j=1,\dots,J$ over the partition $\{B_1,\dots,B_J\}$ of the time interval $\mathcal{I}$ and not the event times $\{S_{i1},S_{i2},\dots\}$. In this case, when the intensity function is smooth and subintervals $B_j$ are narrow enough so that the intensity is approximately constant over $B_j$, one can invoke to the properties of the homogeneous Poisson point processes that conditional on the number of points observed within an interval, unordered locations of points are independent and distributed uniformly over the interval \citep[Theorem 4A]{parzen1999stochastic}. Thus, within each subinterval $B_j,\; j=1,\dots,J$ one can generate $N_i(B_j)$ realizations of a uniform random variable over $B_j$, say $S_{(i,j)} = \{S_{i,j,1},\dots,S_{i,j,N_i(B_j)}\}, j=1,\dots,J$ and form $\Phi_i = S_{(i,1)}\cup S_{(i,2)}\cup \dots\cup S_{(i,J)}$. This type of data is referred as histogram data \citep{streit2010poisson}.

\subsection{Prediction}\label{subsec:Asymptotic_P3ls}
Suppose $\Phi_0$ is a new realization of the log-Gaussian Cox process described in Section \ref{sec:Methodology} and let $X_0$ be its log-intensity function. It is of practical interest to predict the response value $Y_0$ corresponding to the new observation $\Phi_0$. For instance, in our motivating study, an unexperienced radiologist might be interested in an export`s rating of the kidney`s degree of obstruction. By plugging in the estimated coefficient function $\hat{b}(t)$ in the linear model (\ref{eq:Regression_model}), we can predict the response by
\begin{equation}\label{eq:predicted_value}
   \hat{Y}_0 = \bar{Y} + \int_{\mathcal{I}}\left[X_0(t)-\bar{X}(t)\right]\hat{b}_p(t)\;dt.
\end{equation}

Theorem \ref{thm:predictive_consistency} shows the predictive consistency of the predictor (\ref{eq:predicted_value}), by showing that the mean squared prediction error, conditional on the training data, converges to zero as the sample size increases.   

\begin{theorem}\label{thm:predictive_consistency} 
Suppose $(X_0,Y_0)$ is independent of $(X_1,Y_1),\dots,(X_n,Y_n)$ and that $Y_0$ follows model (\ref{eq:Regression_model}). In addition, suppose $B_1,\dots,B_M$, as defined in Section \ref{subsubsec:est_intensity}, are such that $\max_m\{|B|_m\}\leq c_1n^{-1/2}$ and $q$ is large enough such that $\sum_{\ell=q+1}^{\infty}|\xi_{i\ell}|\leq c_2n^{-1/2}$, where $c_1,c_2$ are fixed constants. If $h\to 0$ at a rate to ensure $n^{-1/2}h^{-4}\to 0$ and if we choose $p=p(n)$ to diverge no faster than $n^{1/2}h^{-2}$ and sufficiently slowly to ensure  $n^{-1/2}h^{-4}\lambda^{-1}+n^{-1}h^{-8}\lambda^{-3}\to 0$, as $n\to\infty$, where $\lambda$ is the smallest eigenvalue of the $p\times p$ matrix $H=[h_{jk}]$ with
\[
h_{jk}=\int_\mathcal{I}K^{j+1}(b)(t)K^k(b)(t)\;dt,
\]  
then  
    \begin{equation}
        \left\{\E\left[\left(Y_0-\hat{Y}_0\right)^2 \Bigg| X_1,\dots,X_n\right]\right\}\to 0 \quad\,\mbox{as $n\to\infty$}.
    \end{equation}
   
\end{theorem}

Details of the proof are given in the Supplementary Materials. 

\section{Simulation}\label{sec:Simulation}
In this section we illustrate the performance of the $P^3LS$ algorithm in estimation of the coefficient function as well as in prediction. We compare the proposed method with an alternative functional regression method, as follows. First, we estimate the log-intensities by applying the kernel smoothing method \citep{diggle1985kernel}, i.e. given a point process $\Phi_i$, the intensity function of the process is denoted by $\tilde{\lambda}_i(t)$ and estimated as
\begin{equation} \label{eq:intensity_kernel_est}
    \tilde{\lambda}_i(t) = \sum_{x\in\Phi_i} \frac{\kappa_h(x-t)}{a(t;h)},
\end{equation}
where $\kappa(.)$ is a kernel function, $\kappa_h(u)=\kappa(u/h)/h$, and $a(s;h)=\int\kappa_h(s-x)dx$ is an edge correction term. Then, we apply the functional principal component regression to predict the response. In our numerical investigations, we refer to the alternative method as $FPCR$.

\subsection{Simulation Setup} \label{subsec:Simulation_Setup}

In the simulation studies, we generate 200 realizations of the following random log-intensity functions defined on $\mathcal{I}=[0,24]$.
\begin{equation} \label{eq:simulation_intensity_function}
\log(\lambda_{nobs}(t)) = \sum_{j=1}^{20}\left(\frac{1}{\eta}\omega_j^{(nobs)}+2.8\right)\phi_j(t),
\end{equation}
where $\eta = 10$, $\phi_1,\dots,\phi_{20}$ are B-spline basis functions, and
\begin{align*}
    \omega^{(nobs)}_1&=0, \omega^{(nobs)}_2,\omega^{(nobs)}_{20} \overset{iid}{\sim} N(12,4^2), \omega_j^{(nobs)}\overset{iid}{\sim} N(20,10^2), j=3,\dots,19.
\end{align*}
In addition, we considered 
\[
b(t) = \sum_{k=j}^{20}\vartheta_j\phi_j(t),
\]
for the following four cases of the coefficients $\vartheta_1,\dots,\vartheta_{20}$;
\begin{itemize}
    \item Case 1: $\vartheta_j = \mathbb{I}\{2\leq j\leq 20\}, j=1,\dots,20$.
    \item Case 2: $\vartheta_j = \mathbb{I}\{2\leq j\leq 10\} - \mathbb{I}\{11\leq j\leq 20\}, j=1,\dots,20$.
    \item Case 3: $\vartheta_j = \mathbb{I}\{2\leq j\leq 6\} - \mathbb{I}\{7\leq j\leq 12\} + \mathbb{I}\{13\leq j\leq 20\}, j=1,\dots,20$.
    \item Case 4: $\vartheta_j = -\mathbb{I}\{2\leq j\leq 5\} +  \mathbb{I}\{6\leq j\leq 10\} - \mathbb{I}\{11\leq j\leq 15\} +\mathbb{I}\{16\leq j\leq 20\} , j=1,\dots,20$.
\end{itemize}
Then, $y_1,\dots,y_{200}$ are generated according to the model (\ref{eq:Regression_model}) with $a=0$ and $\epsilon_i\overset{iid}{\sim}N(0,1)$. 
In construction of the coefficient functions we considere examples of a functional relationship where, in Case 1, the response is highly correlated with overall integrated log-intensity ( or equivalently with photon counts) over the experiment`s time; in Case 2, 3 and 4, the response is highly correlated with a contrast in the log-intensity over two, three, and four periods of time. Plots of the coefficient functions considered in Cases 1-4 are illustrated in Figure \ref{fig:Case1-3_Coeff_func}. 

\begin{figure} 
\begin{center}
\begin{tabular}{c}
\includegraphics[width=6.5in, height=3.in]{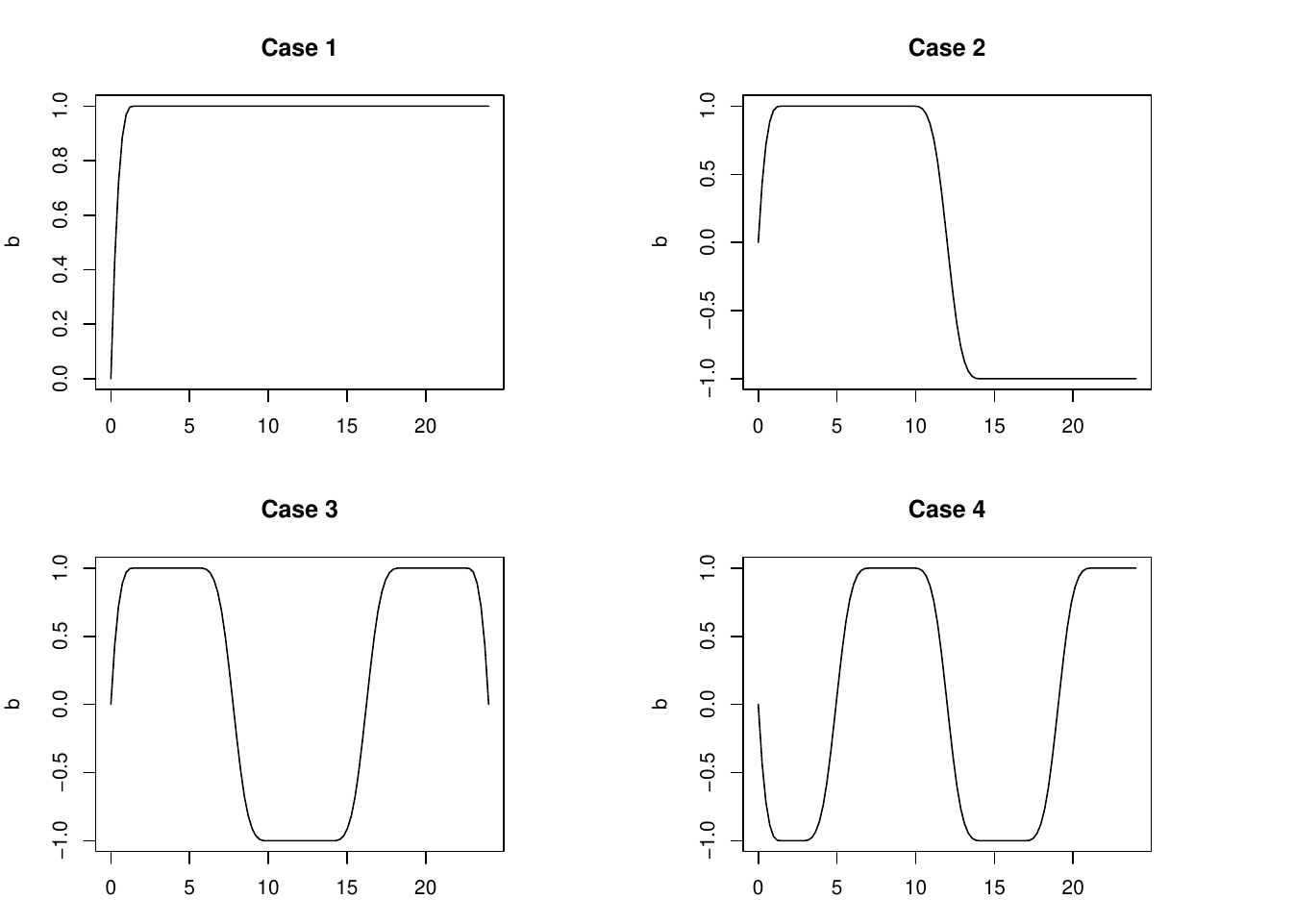}
\end{tabular}
         \caption{Plots of the coefficient functions considered in Case 1 (\textbf{Top Left}), Case 2 (\textbf{Top Right}), Case 3 (\textbf{Bottom Left}), and Case 4 (\textbf{Bottom Right}).}
         \label{fig:Case1-3_Coeff_func}
\end{center}
\end{figure}

To make comparisons, we consider  200 realizations of a temporal Gaussian process generated from the corresponding log-Gaussian Cox process. The randomly selected $n = 100$ samples are used as a training set and another $n_t=100$ samples as a testing set. The mean square estimation error (MSEE) of the coefficient function is defined as
\begin{equation} \label{eq:MSEE}
    MSEE = \int_\mathcal{I}\left[b(t)-b_p(t)\right]^2\;dt,
\end{equation}
and the mean square prediction error ($MSPE$) of the testing responses is defined as
\begin{equation} \label{eq:MSPE}
    MPSE = \frac{1}{n_t}\sum_{j=1}^{n_t}\left(y_j^{(test)}-\hat{y}_j^{(test)}\right)^2,
\end{equation}
where $y_j^{(test)}$ and $\hat{y}_j^{(test)}, j=1,\dots,n_t$ are the response values in the testing set and their predicted values, respectively. The data generation and model fitting procedure are repeated 100 times and boxplots of the root $MSEE$s and root $MSPE$s are computed for $p=1,\dots,10$ basis functions involved in estimation of the coefficient function. Lastly, to apply the $P^3LS$ we choose $q$ large enough such that more than $90\%$ of the total variability is explained by $\phi_1,\dots,\phi_q$, we choose $|B_\ell|=24/100$, and $h=2$.

\subsection{Simulation Results} \label{subsec:Simulation_Results}
Here we display the boxplots of the root $MSPE$ of the estimated model on the testing data set, in Figure \ref{fig:MSPE_sim}, as well as the boxplots of the root $MSEE$  of the coefficient functions, in Figure \ref{fig:MSE_sim}, for $p=1,\dots,10$ basis functions, estimated by the $P^3LS$ and {\it FPCR} methods.
As Figure \ref{fig:MSPE_sim} illustrates, to achieve the $MSPE$ level that $P^3LS$ achieves with $p=1$ basis function, the method based on functional principal component regression requires at least $p=4$ basis functions, in addition that the prediction is more stable. Similar behavior is observed in the estimation of the coefficient functions, as Figure \ref{fig:MSE_sim} illustrates. Finally, we investigate the sensitivity of the prediction performance on the choice of $h$ through cross validation, detailed in the Supplementary Materials. Results indicate that the predictive performance of the $P^3LS$ is not sensitive to the choice of $h$.

\begin{figure}[t] 
\begin{center}
\vspace{-.7in}
\begin{tabular}{cc}
\includegraphics[width=2.9in, height=3.4in, bb = 50 50 650 650]{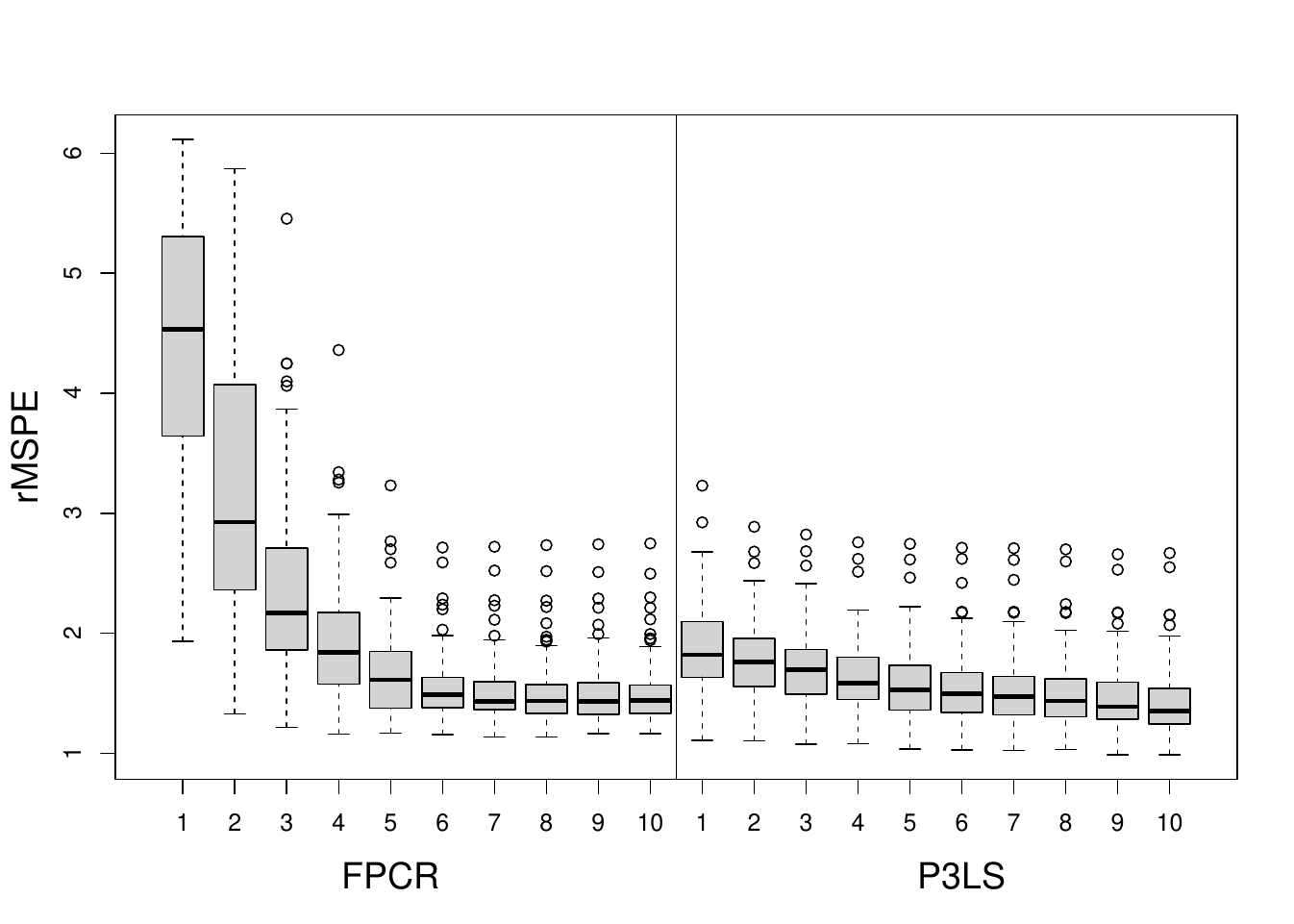} 
& \includegraphics[width=2.9in, height=3.4in, bb = 50 50 650 650]{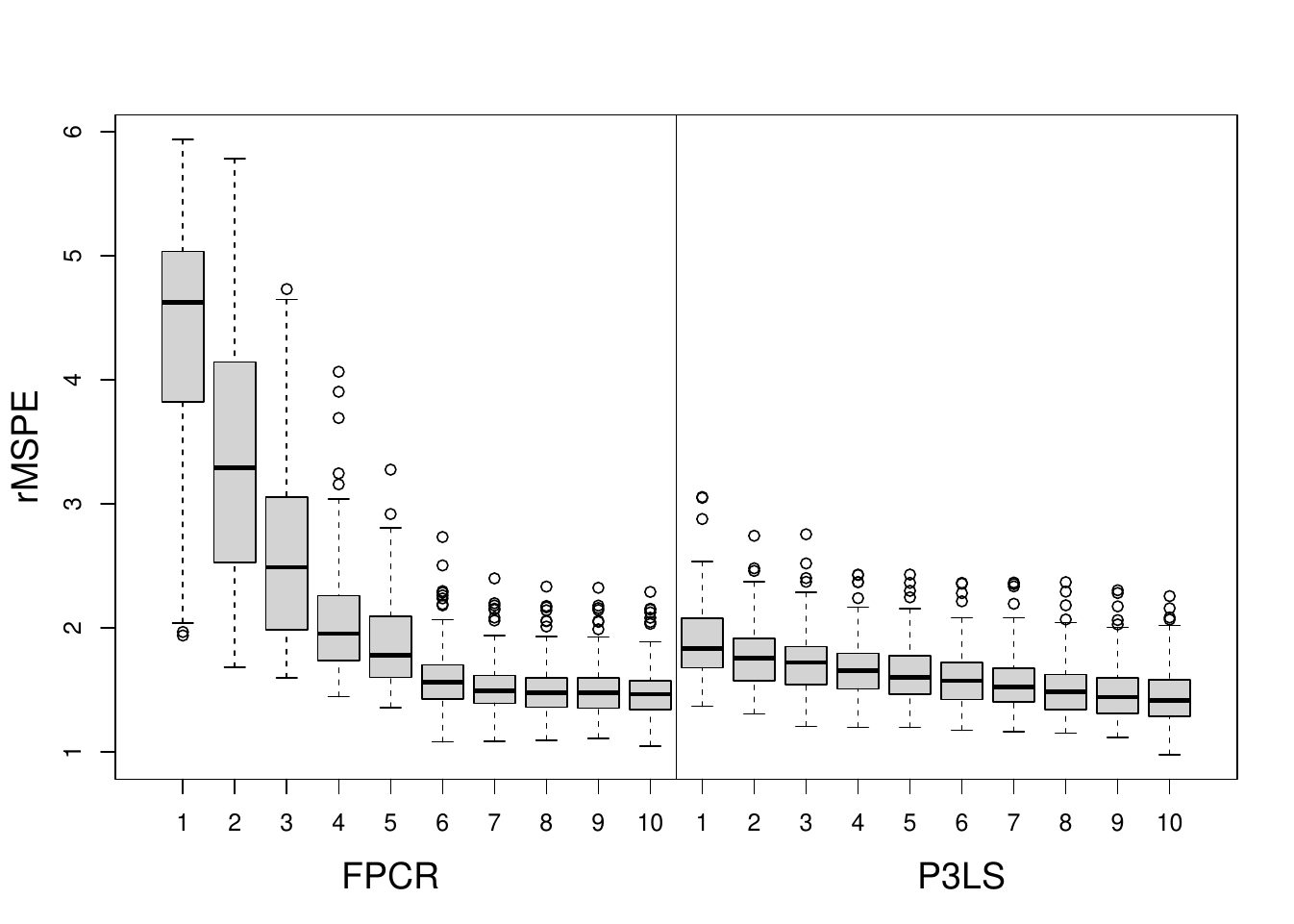}\vspace{-67pt} \\
\includegraphics[width=2.9in, height=3.4in, bb = 50 50 650 650]{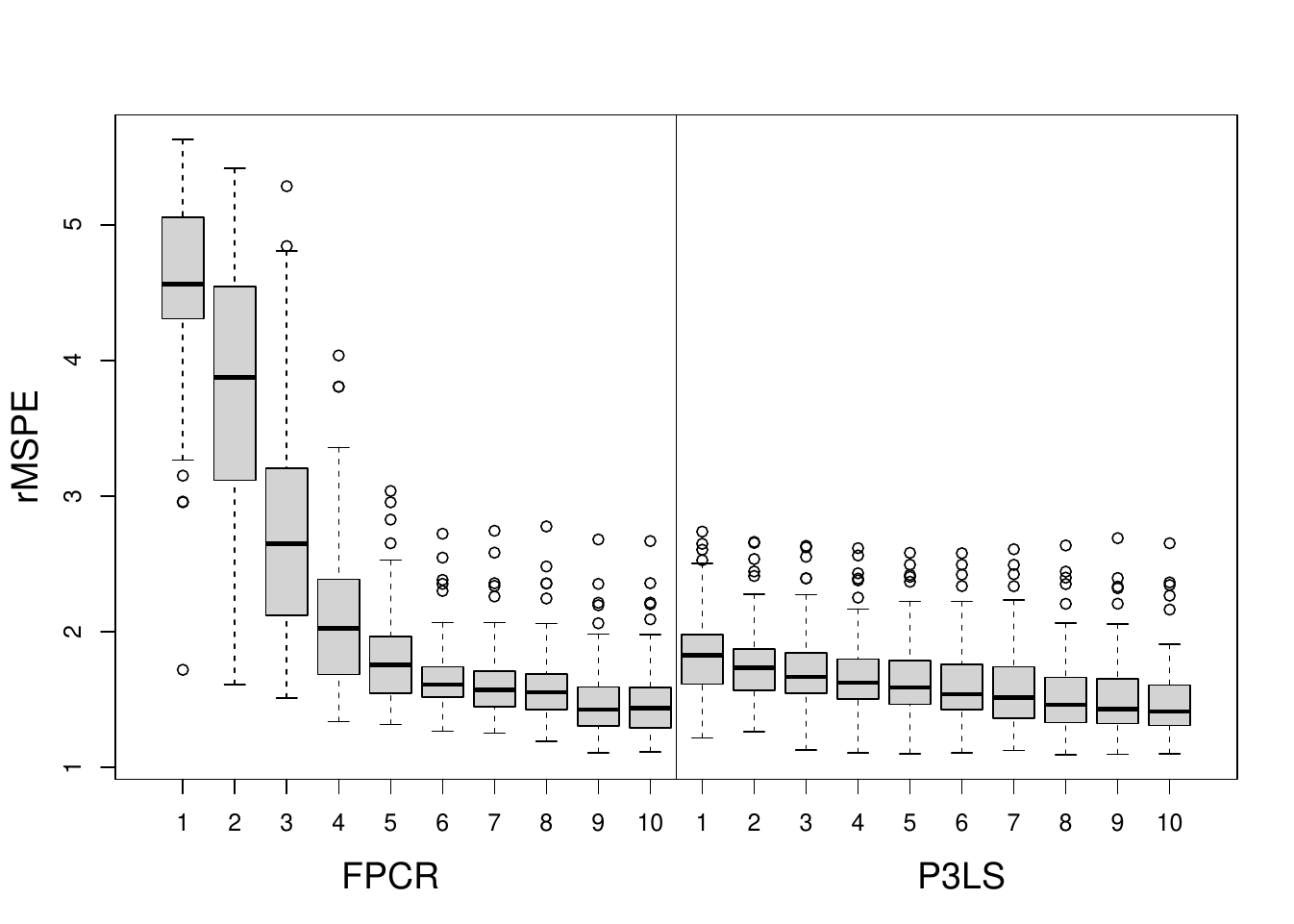}
& \includegraphics[width=2.9in, height=3.4in, bb = 50 50 650 650]{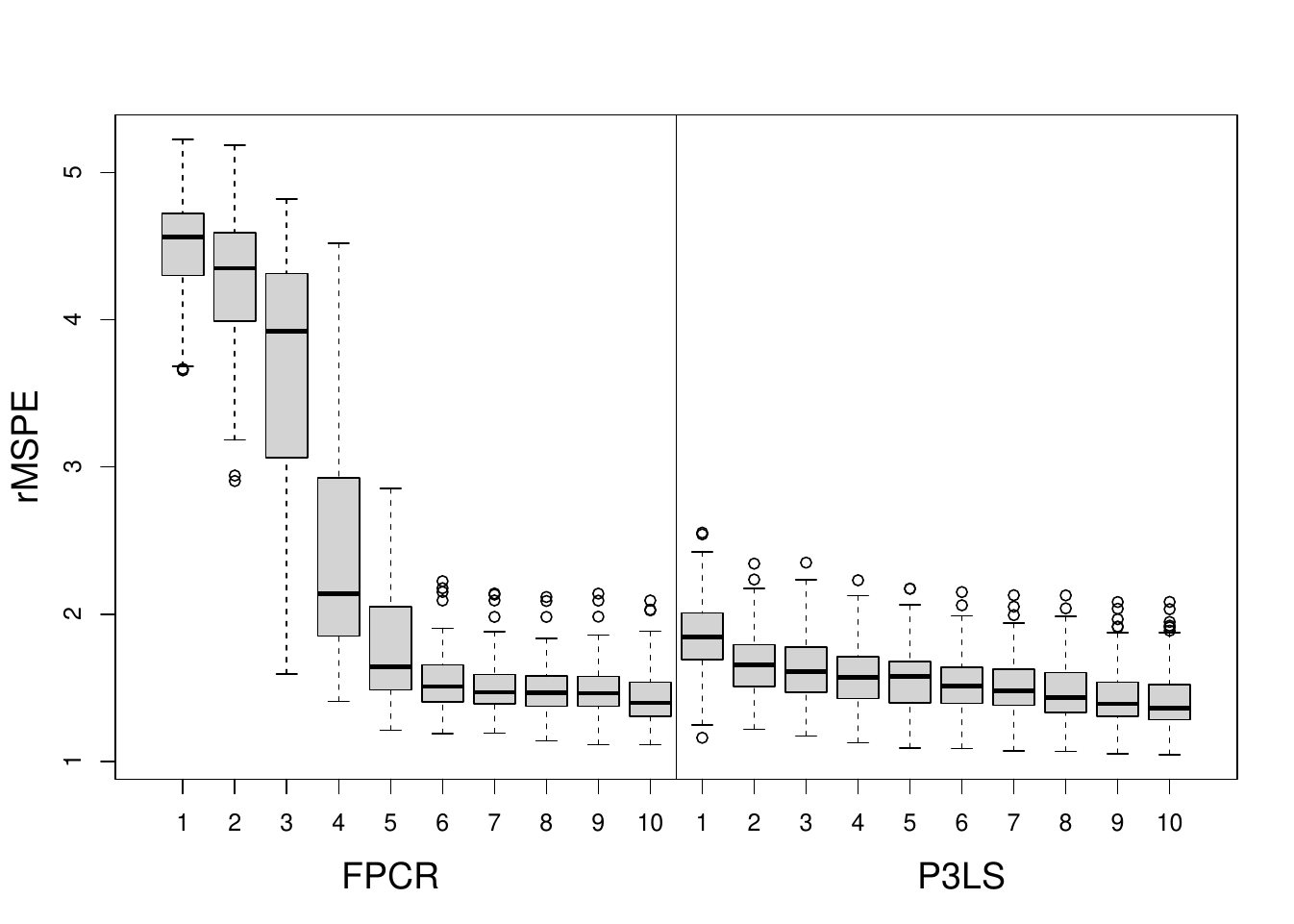}
\end{tabular}
         \caption{\textbf{Top Left:} Root $MSPE$ of Case 1, \textbf{Top Right:} Root $MSPE$ of Case 2, \textbf{Bottom Left:} Root $MSPE$ of Case 3, \textbf{Bottom Right:} Root $MSPE$ of Case 4.}
         \label{fig:MSPE_sim}
\end{center}
\end{figure}

\begin{figure}[t] 
\begin{center}
\vspace{-.7in}
\begin{tabular}{cc}
\includegraphics[width=2.9in, height=3.4in, bb = 50 50 650 650]{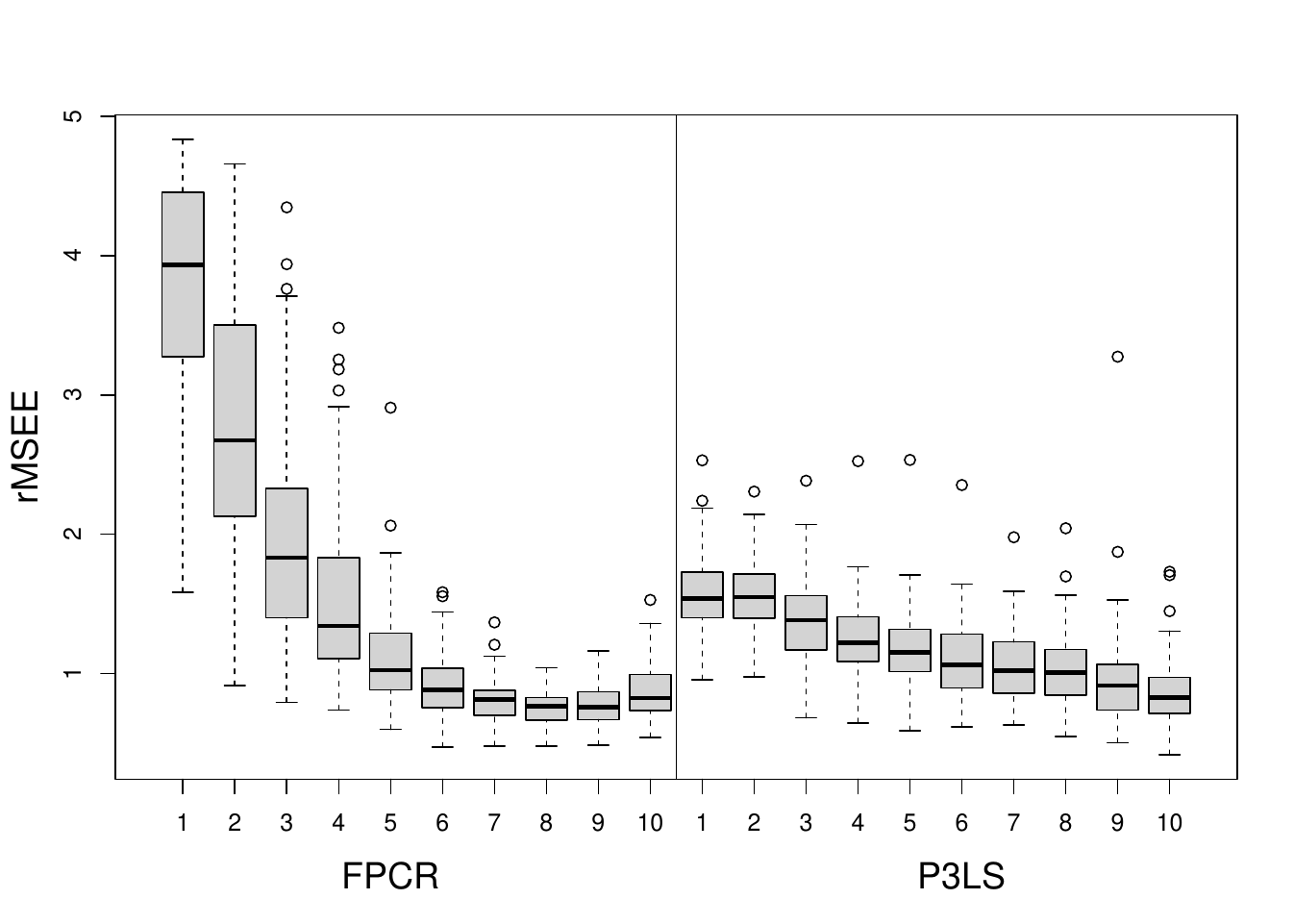} 
& \includegraphics[width=2.9in, height=3.4in, bb = 50 50 650 650]{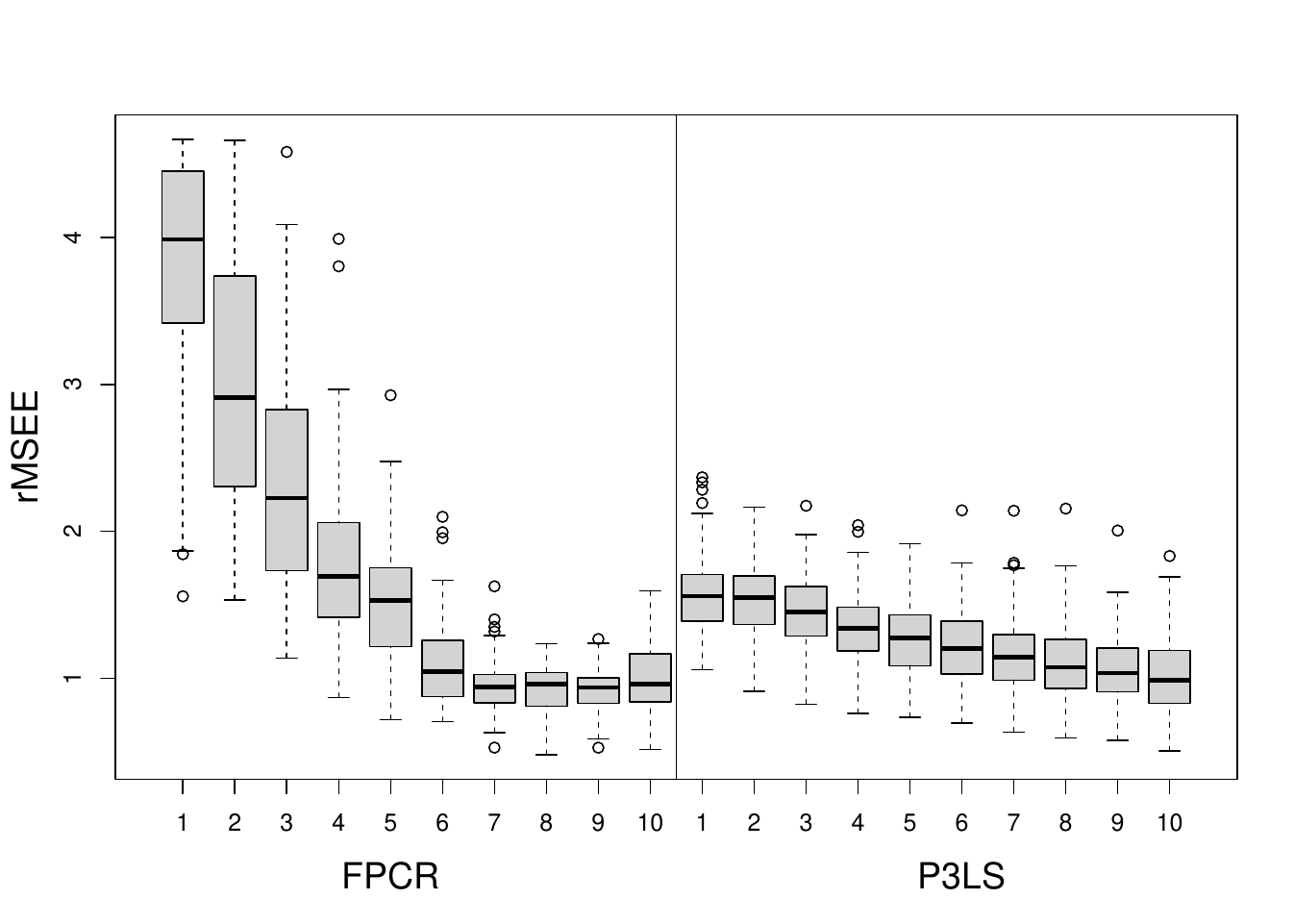}\vspace{-67pt} \\
\includegraphics[width=2.9in, height=3.4in, bb = 50 50 650 650]{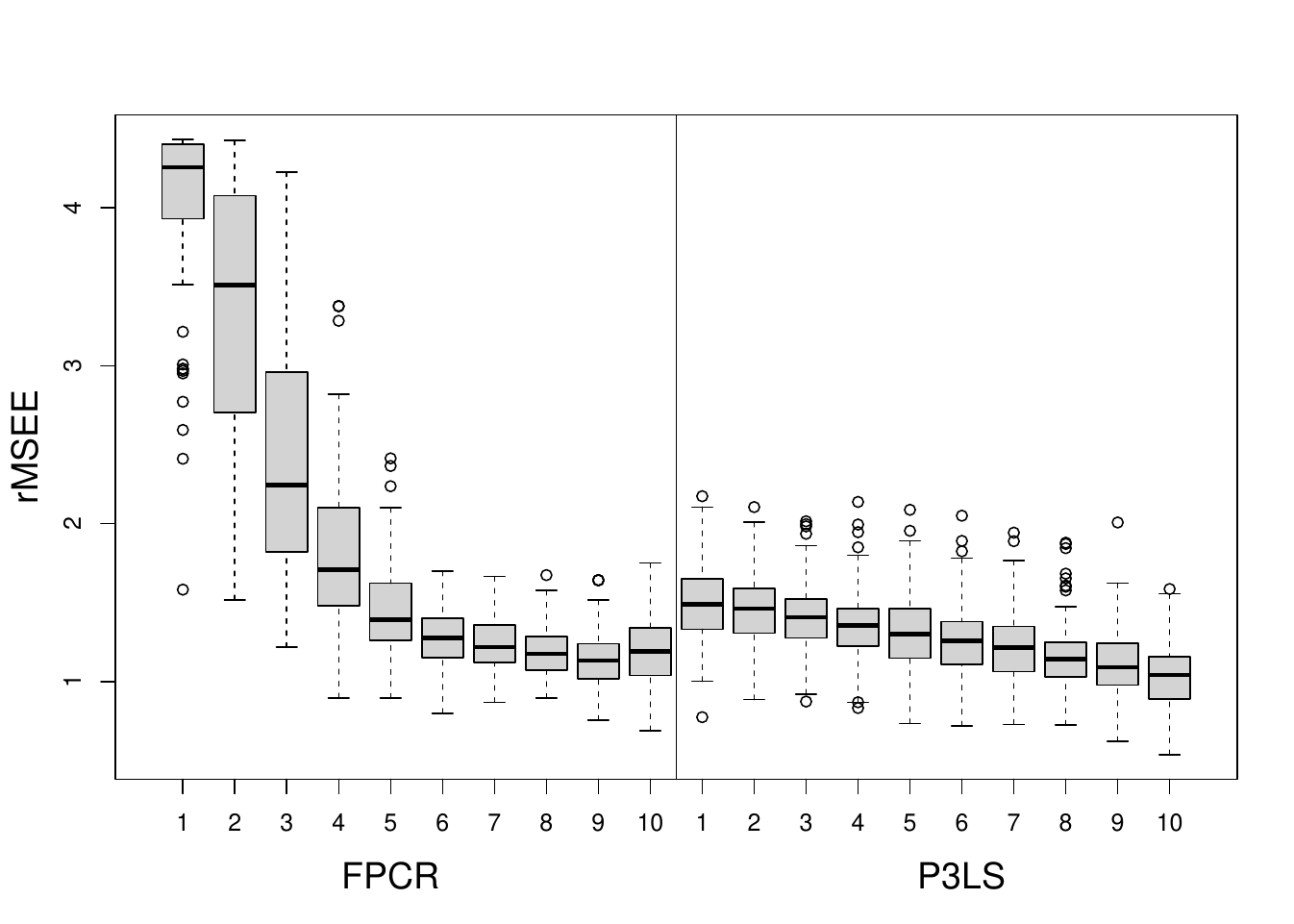}
& \includegraphics[width=2.9in, height=3.4in, bb = 50 50 650 650]{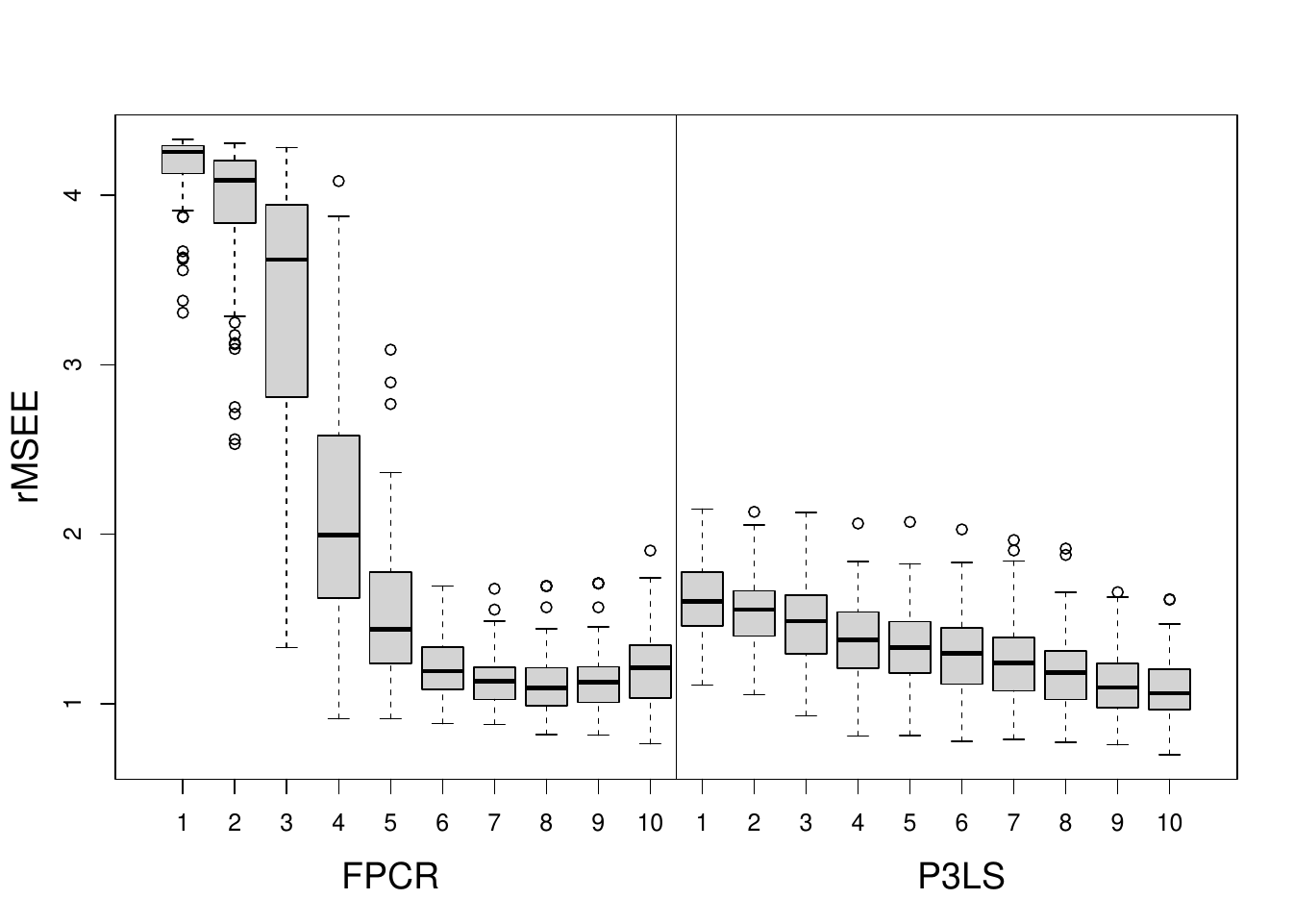}
\end{tabular}
         \caption{\textbf{Top Left:} Root $MSEE$ of Case 1, \textbf{Top Right:} Root $MSEE$ of Case 2, \textbf{Bottom Left:} Root $MSEE$ of Case 3, \textbf{Bottom Right:} Root $MSEE$ of Case 4.}
         \label{fig:MSE_sim}
\end{center}
\end{figure}


\section{Radionuclide Imaging  Study} \label{sec:Data_Analysis}
We consider a renal radionuclide imaging  study conducted at  Emory University as part of an effort in developing analytical tools that help radiologists with the interpretation  of kidney obstruction. Kidney obstruction refers to the condition in which there is a blockage in the ureter that can have a degenerative impact on kidney function and can result in kidney failure, if not treated in a timely manner. A widely used approach for evaluating suspected obstruction is radionuclide imaging. Imaging is performed following an intravenous injection of the gamma emitting tracer technetium-99m mercaptoacetyltriglycine (Tc-99mMAG3), with additional imaging following the subsequent intravenous administration of a potent diuretic \citep{o1996consensus}. Lack of opportunity and insufficient training, however, can result in scan interpretations by less experienced  radiologists that disagree considerably with each other and  disagree with the experts' interpretations \citep{jaksic2005, Taylor216, TAYLOR201241}. This highlights the need for analytic tools to support the interpretation of kidney obstruction. Such tools could improve patient care by reducing both intra- and inter-observer variability in MAG3 scan evaluations and by enhancing the training of radiology residents with limited experience.


The data considered in this paper were collected from $N=131$ patients, during the period of March 1998 to July 2017 who were referred to the clinic
with suspected kidney obstruction. Scans of both kidneys were available for $122$ of them and renogram data of only one kidney was available for $9$ subjects, with data available for $127$ left and $126$ right kidneys. Of those 131 subjects, 66 were female and 65 were male with median age of 59 where 75\% of them were between 48 to 70 years old.

Each subject underwent two scans, the first scan called ``Baseline" and a second scan following an injection of furosemide, a diuretic, ``Diuretic". The Baseline scan was performed following the intravenous injection of MAG3, which is rapidly removed from the blood by the kidneys and then travels down the ureters from the kidney to the bladder. Photons emitted by the tracer were
imaged by a gamma camera/computer system and quantified
for analysis by placing a region of interest (ROI) over each
kidney. The point process of photon arrival times for the baseline scan was obtained from the whole-kidney ROIs over a 24-minute acquisition. Following intravenous administration of furosemide, a second (diuretic) scan was performed, with data collected for an additional 20 minutes, yielding both baseline and diuretic renograms.
{In this study, there is no established gold standard for assessing kidney obstruction.} Therefore, an expert with extensive knowledge of kidney function and over 25 years of experience in academic nuclear medicine was asked to interpret each kidney's condition on a scale from -1 to 1, where values approaching 1 indicate a high degree of obstruction. 

In our analysis, we used data from 100 kidneys as training set and evaluated the trained model using the remaining data. 
The proposed method described in Section \ref{sec:Methodology} was applied separately to data from the left and right kidneys, with the baseline and diuretic scans concatenated to jointly incorporate information from both processes.

To fit the predictive model (\ref{eq:Regression_model}), we first evaluated several candidate basis functions and, guided by the Bayesian Information Criterion, selected two. We then estimated both the basis functions and the coefficient function in (\ref{eq:Regression_model}) using the training dataset.
The top panel of Figure (\ref{fig:Coeff_Basis_MSPE}) illustrates the estimated coefficient function of the predictive model (\ref{eq:Regression_model}) for both the left and right kidneys. For both kidneys, the estimated coefficient function is negative across all times for the baseline renogram and positive across all times for the diuretic renogram. It represents a contrast between diuretic and baseline renograms such that larger increases in diuretic log-intensities relative to baseline are associated with higher expected expert scores. This association 
is consistent with expert clinical knowledge, where larger values of the expert score are associated with obstruction and where the renogram  of an obstructed kidney tends to increase during the baseline scan and stay at the same or higher level during the diuretic scan.

\begin{figure}[h] 
\begin{center}
\begin{tabular}{cc}
\vspace{-50pt}

\includegraphics[width=2.7in, height=2.7in, bb = 50 50 600 600]{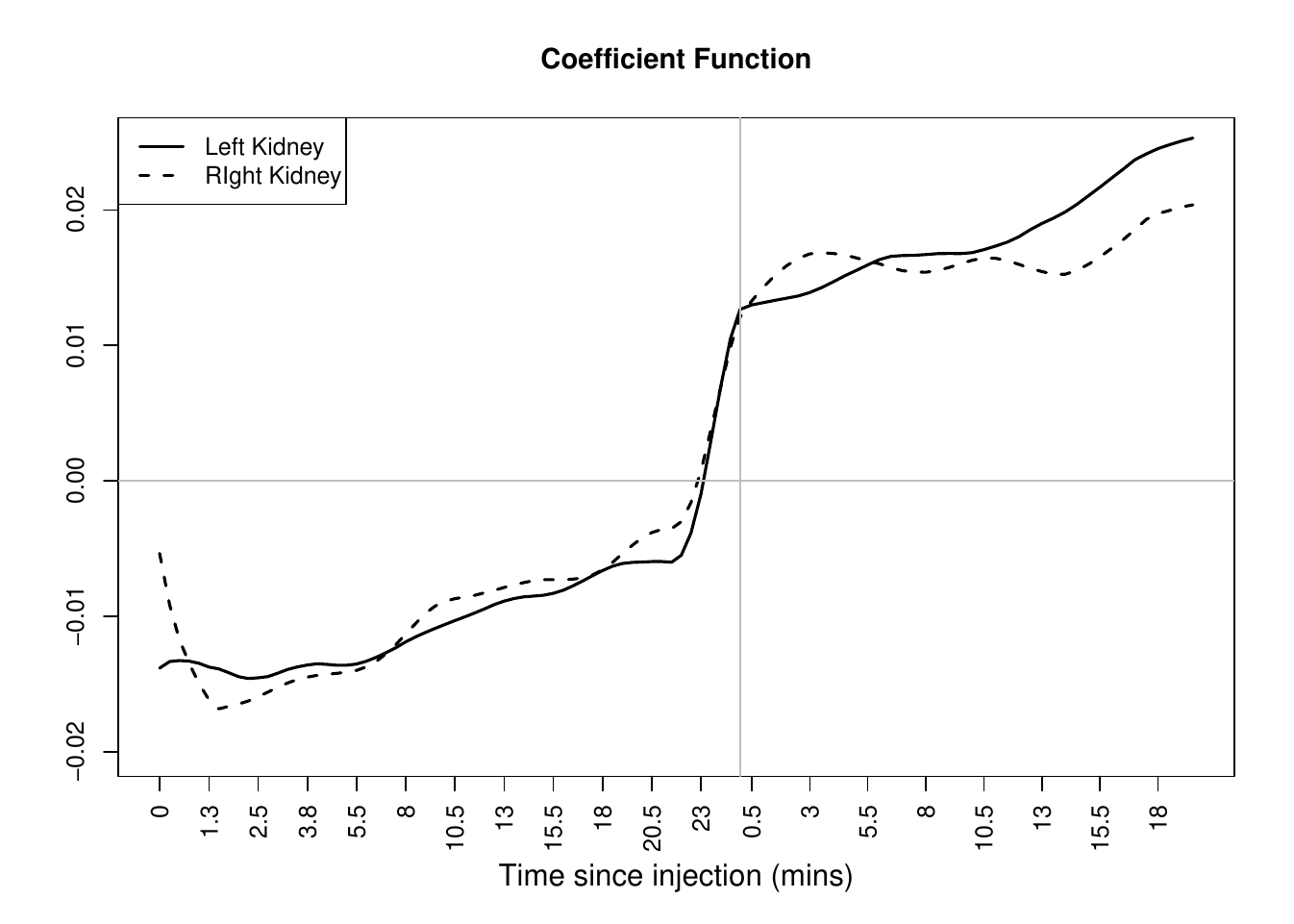}& \\
\includegraphics[width=2.7in, height=2.7in, bb = 50 50 600 600]{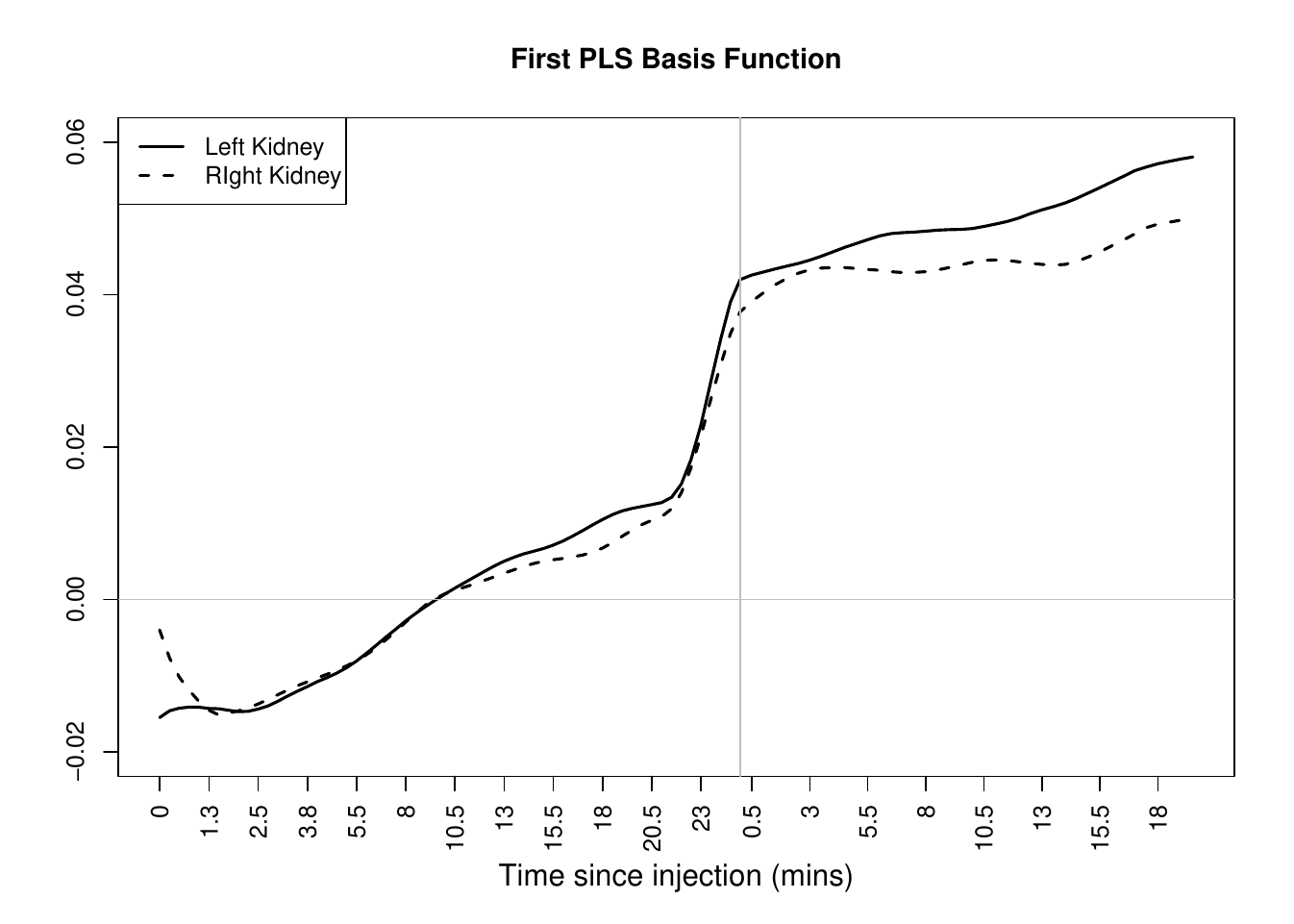}
& \hspace{10pt}
\includegraphics[width=2.7in, height=2.7in, bb = 50 50 600 600]{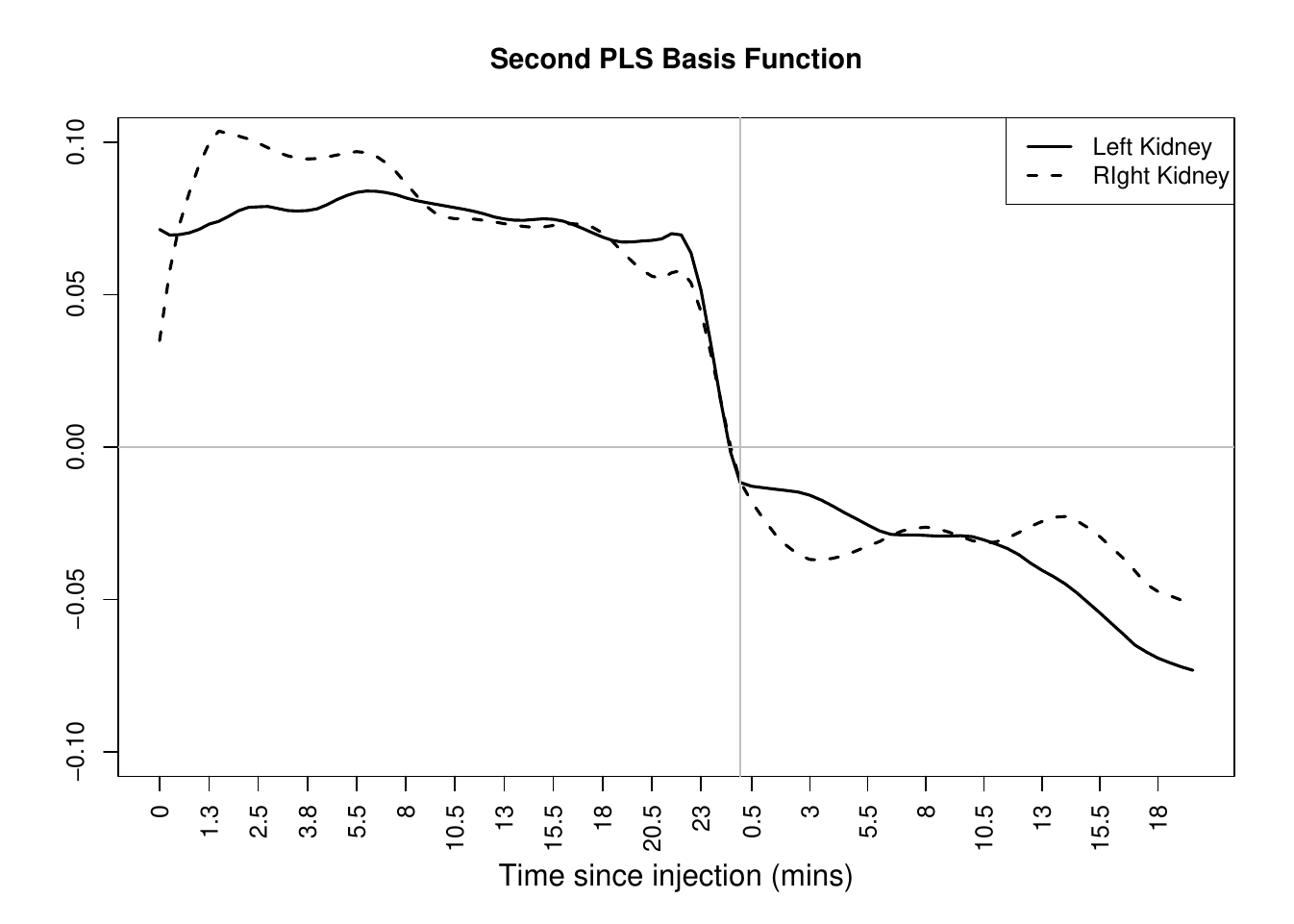} 
\end{tabular} 
\vspace{20pt}
         \caption{\textbf{Top Left panel} The coefficient functions estimated for the left (solid line) and right (dashed line) kidneys. \textbf{Bottom Left panel} The first PLS basis function of the left (solid line) and right (dashed line) kidney; \textbf{Bottom right pane}: The second PLS basis function of the left (solid line) and right (dashed line) kidney.}
    \label{fig:Coeff_Basis_MSPE}
\end{center}
\end{figure}

The estimated two basis functions, which are displayed in the middle panel of Figure \ref{fig:Coeff_Basis_MSPE}, also represent contrasts and provide insight into parsimonious temporal renogram information that is predictive of obstruction.  The first basis function is negative for early baseline times before 9 minutes, and positive for both late baseline times after 9 minutes as well as for all diuretic scans. 
The second basis function is positive at all baseline time points and negative at all diuretic time points. These findings illustrate one of the main benefits of $P^3LS$ compared to approaches such as the two-stage FPCA: it provides parsimonious, potentially clinically interpretable, one-dimensional measures of optimal association between variability in log-intensity with clinical outcomes. In our study, we found that the speed of uptake in intensity at the beginning of baseline relative to intensity during the rest of the study (first basis function) and intensity of the diuretic renogram relative to baseline (second basis function), which are clinically meaningful insights, are most predictive of obstruction score.    

Software to perform our method and examples are included to demonstrate its utility. See Web Appendix.


\section{Discussion} \label{sec:Discussion}
This article introduced, to the best of our knowledge, the first extension of partial least squares to point process data. We have explored its performance in comparison with an intuitive approach to linear prediction with point process covariates in various settings of practical interest. The advantages of the method over  alternative approaches to linear prediction with log-Gaussian Cox processes as covariates stem from incorporation of the properties of the process in estimation of the covariance function as well as the log-intensity functions. The incorporation of this information provides a more efficient dimension reduction compared to the use of prespecified basis functions, accounting for more variability with few basis functions. This parsimonious set of basis functions also has the benefit of providing a collection of potentially scientifically interpretable one-dimensional measures of the functional process that are most predictive of the outcome.

The method is not exhaustive and can be extended to more complicated scenarios. In our motivating study, we analyzed the left and right kidneys separately, despite the fact that the data contain scans of both kidneys for 122 of the patients. An optimal analysis of the data needs to account for the dependence within and between the levels in a multilevel data, where an extension of the PLS is needed. The second extension focuses on the incorporation of space-time point processes. In many medical image studies in nuclear medicine, two-dimensional images are produced using gamma camera for each subject.
Extension of the method to higher-dimensional point processes can potentially be of interest. Finally, scalar covariates are also present in addition to the point process covariates in some situations. One possible approach to incorporate scalar covariates in the predictive model is a two-step procedure of first computing the PLS basis functions using the residuals after removing the linear effect of the scalar covariates. Second,  estimate the coefficients of the scalar covariates as well as the coefficient function of the linear model for point process using the computed PLS basis functions. We leave the extensions of the method to incorporate scalar covariates as well as multilevel and multi-dimensional point processes for future investigations. 

\section*{Acknowledgment}
We thank Dr. Andrew Taylor at the department of nuclear medicine, Emory university, for informative discussions with the renal study. This work is supported by National Institutes of Health grants R01GM140476, R01HL159213 and R01MH125816.

\bibliographystyle{asa}

\bibliography{P3LS}

\newpage

\begin{appendices}

\begin{center}
{\large\bf Supplementary Materials for $\mbf{P^3LS}$: Point Process Partial Least Squares}
\end{center}
\appendix

\begin{abstract}
This document provides appendices to supplement the information presented in the main manuscript. In \ref{appendix:Gram-Schmidt} the modified Gram-Schmidt Algorithm is provided. \ref{app:Asymptotic_analysis} provides the details of the proof of Theorem 1 of the main manuscript. Additional results and illustrations from the simulation studies are presented in  \ref{appendix:extra_simulation}.
\end{abstract}

\section{Modified Gram-Schmidt Algorithm}\label{appendix:Gram-Schmidt}
This algorithm constructs a set of unit length orthogonal basis functions $u_1,\dots,u_p$ from a set of linearly independent functions $v_1,\dots,v_p$, where orthogonality is defined with respect to the inner product $<.,.>$. For two functions $f_1$ and $f_2$, the inner product is defined as
\begin{equation}
    <f_1,f_2> = \int_\mathcal{I}\int_\mathcal{I}f_1(s)\hat{K}(s,t)f_2(t)\;ds\;dt.
\end{equation}
The modified Gram-Schmidt algorithm is as follows: \\

\begin{algorithm}[H]
\caption{Modified Gram-Schmidt} 

\SetKwInput{KwInput}{Input}                
\SetKwInput{KwOutput}{Output}              
\DontPrintSemicolon
\SetAlgoLined
\RestyleAlgo{boxruled}
  
  \KwInput{Set of linearly independent functions $v_1,\dots,v_p$}
  \KwOutput{Set of orthogonal functions $u_1,\dots,u_p$}

\BlankLine
    \For{$j\in\{1,\dots,p\}$}{
    $u_j^{[1]} = v_j$\;
    For{$i=1,\dots,j-1$}{
    $u_j^{[i+1]} = u_j^{[i]} - <u_j^{[i]},u_i,>u_i$
    }\;
    $u_j=u_j^{[j]}/\|u_j^{[j]}\|$
    }
  \KwOutput{$u_1,\dots,u_p$}
\end{algorithm}
\BlankLine

\section{Asymptotic Analysis of $P^3LS$}
\label{app:Asymptotic_analysis}
To study the asymptotic properties of the $P^3LS$ we follow the analysis detailed in \cite{Delaigle_Hall_2012} and remind that the  the space spanned by $\psi_1,\dots,\psi_p$, for each $p\geq 1$, is the same as the space spanned by $K(b),\dots, K^p(b)$. This property of partial least squares enables us to study its asymptotic properties by employing function expansion in the $K(b),\dots,K^p(b)$.
   
Let $X(t), X_i(t)\in\mathcal{C}(\mathcal{I}), i=1,\dots,n$, be i.i.d Lipchitz random functions with Lipchitz constants at most $c_0$ such that $\int_{\mathcal{I}}\E[X_i^2(t)]\;dt < \infty$ and $\E[\|X_i^4\|]<\infty$. In addition, lets assume that $\int_{\mathcal{I}}b^2(t)\; dt < \infty$ and $\E[\epsilon_i]<\infty$. Recall that, $g(x) = \E[Y|X=x] = a+\int_\mathcal{I}b(t)x(t)\;dt$. Note that, according to Theorem 3.2 in \cite{Delaigle_Hall_2012}, when all eigenvalues of $K(.,.)$ are nonzero, any $b\in\mathcal
{C}(\mathcal{I})$ has the following representation 
\begin{equation} \label{eq:expansion_of_b_in_K^j}
    b(t) = \sum_{j=1}^{\infty}\gamma_jK^j(b)(t),
\end{equation}
where the series converges in $L^2$. We approximated $b(t)$
by its truncated expansion in the PLS basis functions $K(b),\dots,K^p(b)$, denoted as $b_p(t)$ i.e. $b_p(t) = \sum_{j=1}^{p}\gamma_jK^j(b)$. As outlined in \cite{Delaigle_Hall_2012}, $\gamma_1,\dots,\gamma_p$ is defined to be the sequence $w_1,\dots,w_p$ that minimizes 
\begin{align}\label{eq:estimation_Error}
    t_p(w_1,\dots,w_p) = \E\left\{\int(X-\E(X))b-\sum_{j=1}^{p}w_j\int(X-\E(X))K^j(b)\right\}^2,
\end{align}
where, in matrix notation, $\gamma=[\gamma_1,\dots,\gamma_p]^\prime$ can be written as 
\begin{equation}
    \gamma =H^{-1}\alpha
\end{equation}
and $\alpha=[\alpha_1,\dots,\alpha_p]$ and $H=[h_{jk}]_{1\leq j,k\leq p}$ are defined as follows
\begin{align*}
    h_{jk} &= \int K^{j+1}(b)K^k(b), \\
    \alpha_j &= \int K(b) K^j(b) = h_{0,j}.
\end{align*}
This gives rise to approximation of $g(x)$ by $g_p(x)$, where $g_p(x) = a+\int_\mathcal{I}b_p(t)x(t)\;dt = \E[Y]+\sum_{j=1}^{p}\gamma_j\int_{\mathcal{I}}(x-\E[X])K^j(b)(t)\;dt$. Next to estimate $b(t)$, we first estimate $K(s,t)$ by $\hat{K}(s,t)$ and $K(b),\dots,K^j(b)$ by $\hat{K}(b),\dots,\hat{K}^p(b)$, where we used the predicted log-intensities in order to compute $\hat{K}(b)$. Then, we estimate $\gamma$ by $\hat{\gamma}=\hat{H}^{-1}\hat{\alpha}$ where, $\hat{H}=[\hat{h}_{jk}]_{1\leq j,k\leq p}$
\begin{align*}
    \hat{h}_{jk} &= \int \hat{K}^{j+1}(b)\hat{K}^k(b), \\
    \hat{\alpha}_j &= \int \hat{K}(b) \hat{K}^j(b) = \hat{h}_{0,j},
\end{align*}
and estimate $g_p(x)$ by $\hat{g}_p(x)=\bar{Y}+\sum_{j=1}^{p}\hat{\gamma}_j\int_{\mathcal{I}}(x-\E[X])\hat{K}^j(b)(t)\;dt$.

Consider $X_0$, having the same distribution as $X_1,\dots,X_n$ and being independent of them. Let $\|.\|_{\pred}$ denote the predictive norm, conditional on $X_1,\dots,X_n$; i.e. if $W$ is a random variable, then $\|W\|_{\pred} = \{\E(W^2|X_1,\dots,X_n)\}^{1/2}$. The following theorem establishes predictive consistency and the convergence rate of $\hat{g}(X_0)$ to $g(X_0)$ as $n\to\infty$.

\begin{theorem}
    Under conditions of Theorem \ref{thm:predictive_consistency}, if $h\to 0$ at a rate to ensure $n^{-1/2}h^{-4}\to 0$ and if we choose $p=p(n)$ to diverge no faster than $n^{1/2}h^{-2}$, sufficiently slowly to ensure  $n^{-1/2}h^{-4}\lambda^{-1}+n^{-1}h^{-8}\lambda^{-3}\to 0$, as $n\to\infty$, then 
    \begin{align*}
        \|\hat{g}_p(X_0)-g_p(X_0)\|_{\pred} &=  O_p\left\{n^{-1/2}h^{-4}\lambda^{-1}+n^{-1}h^{-8}\lambda^{-3} + 
        t_p(\gamma_1,\dots,\gamma_p)^{1/2}\right\}.
    \end{align*}
\end{theorem}

    Note that the term in (\ref{eq:estimation_Error}) converges to zero as $p\to\infty$. This is a consequence of the convergence in $L^2$ of the series in (\ref{eq:expansion_of_b_in_K^j}), that is established in Theorem 3.2 of \cite{Delaigle_Hall_2012}. Therefore, the $P^3LS$ achieves predictive consistency provided that $h\to 0$ and $p\to\infty$ at a rate specified in Theorem \ref{thm:predictive_consistency}.

We proceed by establishing the convergence rate of $\hat{g}(X_0)$ to $g(X_0)$ by deriving
\begin{enumerate}
    \item[I.] the convergence rate of $\hat{K}$ to $K$ 
    \item[II.] the convergence rate of $\hat{K}^j(b)$ to $K(b)$
    \item[III.] the convergence rate of matrix entries $\hat{h}_{jk}$ to $h_{jk}$
    \item[IV.] the convergence rate of $\hat{\gamma}$ to $\gamma$
    \item[V.] the asymptotic expansion of $\hat{g}_p(x)-g_p(x)$
    \item[VI.] the convergence rate of $\hat{g}(X_0)$ to $g(X_0)$
\end{enumerate}

\subsection*{I. Convergence rate of $\hat{K}$ to $K$}
\begin{proposition}
\begin{align}
    \hat{K} = K + n^{-1/2}h^{-2}\zeta,
\end{align}
    where $\zeta=O_p(1)$.
\end{proposition}
\begin{proof}
Proof follows along the same line as in the proof of Theorem 1 in \cite{xu2020semi}. 
\end{proof}

\subsection*{II. Convergence rate of $\hat{K}^j(b)$ to $K(b)$}

\begin{proposition}\label{remark:X_conditions}
     If we choose $B_\ell$`s such that $\max_k\{|B|_k\}\leq c_1n^{-1/2}$ and $\sum_{\ell=q+1}^{\infty}|\xi_{i\ell}|\leq c_2n^{-1/2}$, where $c_1,c_2$ are fixed constants and $\xi_{i\ell}$`s are as defined in equation (14) of the main manuscript, then $\|X_i(t)-\hat{X}_i^{(q)}(\bar{t}_k)\|=O_p\{n^{-1/2}h^{-2}\}$.
\end{proposition}

\begin{proof} 
    Recall that the log-intensity functions $X_i(t)$ were estimated by $\hat{X}_i^{(q)}(t) = \sum_{\ell=1}^{q}\hat{\xi}_{i\ell}\hat{\phi}_\ell(t)$, where $\hat{\phi}_1,\hat{\phi}_2,\dots$ are estimated eigen functions of the covraiance function $K(.,.)$. In addition, recall that $\bar{t}_k$ is the mid-bin-point of $B_k$. Therefore, we have the following decomposition of $\|X_i(t)-\hat{X}_i^{(q)}(\bar{t}_k)\|$
    \begin{align*}
        \|X_i(t)-\hat{X}_i^{(q)}(t)\| &\leq \|X_i(t)-X_i^{(q)}(t)\| + \|X_i^{(q)}(t)-X_i^{(q)}(\bar{t}_k)\| \;+ \\
        &\quad\; \|X_i^{(q)}(\bar{t}_k)-\hat{X}_i^{(q)}(\bar{t}_k)\| + \|\hat{X}_i^{(q)}(\bar{t}_k)-\hat{X}_i^{(q)}(t)\|.
    \end{align*}
    Let $C$ be the Lipchitz constant of $X_i$ then $\|X_i^{(q)}(t)-X_i^{(q)}(\bar{t}_k)\|\leq C|B_k|$. In addition, $\|X_i(t)-X_i^{(q)}(t)\|\leq \sum_{\ell=q+1}^{\infty}|\xi_{i\ell}|$ and $\|X_i^{(q)}(\bar{t}_k)-\hat{X}_i^{(q)}(\bar{t}_k)\|=O_p\{n^{-1/2}\}$, since $\xi_{i\ell}$`s were estimated by the maximum likelihood method. In addition, $\|\hat{X}_i^{(q)}(\bar{t}_k)-\hat{X}_i^{(q)}(t)\|=O_p\{n^{-1/2}h^{-2}\}$, since
    \begin{align*}
        \hat{\phi}(\bar{t}_k)-\hat{\phi}(t) = \hat{\phi}(\bar{t}_k)-\phi(\bar{t}_k) + \phi(\bar{t}_k)-\phi(t)+\phi(t)-\hat{\phi}(t),
    \end{align*}
    and $\|\hat{\phi}(\bar{t}_k)-\phi(\bar{t}_k)\|=O_p\{n^{-1/2}h^{-2}\}$, $\|\phi(\bar{t}_k)-\phi(t)\|\leq C|B_k|$, $\phi(t)-\hat{\phi}(t)=O_p\{n^{-1/2}h^{-2}\}$.
    This implies that 
    \begin{align*}
        \|X_i(t)-\hat{X}_i^{(q)}(\bar{t}_k)\| = O_p\left\{n^{-1/2}+\max_k\{|B|_k\}+\sum_{\ell=q+1}^{\infty}|\xi_{i\ell}|+n^{-1/2}h^{-2}\right\}
    \end{align*}
    Thus, if we choose $B_\ell$`s such that $\max_k\{|B|_k\}\leq c_1n^{-1/2}$ and $\sum_{\ell=q+1}^{\infty}|\xi_{i\ell}|\leq c_2n^{-1/2}$, where $c_1,c_2$ are fixed constants, then $\|X_i(t)-\hat{X}_i^{(q)}(\bar{t}_k)\|=O_p\{n^{-1/2}h^{-2}\}$.
\end{proof}

\begin{proposition}
If $n^{-1/2}h^{-4}\to 0$ as $n\to\infty$ and $h\to 0$, then
    \begin{align}
        \hat{K}^j(b) = K^j(b)+n^{-1/2}\zeta_j + n^{-1}\eta_j,
    \end{align}
where $\|\zeta_j\| = O_p(h^{-2}j\|K\|^j)$ and $\|\eta_j\| = O_p(h^{-4}\|K\|^j)$.
\end{proposition}
\begin{proof}
    
\begin{align*}
    \hat{K}(b) &= \frac{1}{n}\sum_{i=1}^{n} \left( \hat{X}_i^{(q)}(t) - \bar{\hat{X}}_i^{(q)}(t)\right)(Y_i-\bar{Y}) - \E\left[\{X_i(t)-\mu(t)\}(Y_i-\mu_Y)\right] \\
    &= \frac{1}{n} \sum_{i=1}^{n} (X_i(t)-\bar{X}(t))(Y_i-\bar{Y}) - \E\left[\{X_i(t)-\mu(t)\}(Y_i-\mu_Y)\right] \; + \\
    &\qquad\qquad\frac{1}{n}\sum_{i=1}^{n}\left(\hat{X}_i^{(q)}-X_i(t)+\bar{X}(t) - \bar{\hat{X}}^{(q)}(t)\right)(Y_i-\bar{Y}) \\
    &= K(b) + n^{-1/2}\tilde{\zeta}_0 + n^{-1}(\tilde{\eta}_0+\rho_0), 
\end{align*}
where 
\begin{align*}
    \tilde{\zeta}_0 &= n^{-1/2}\sum_{j=1}^{n}(1-\E)\{X_i(t)-\mu(t)\}\{Y_i-\E(Y_i)\}, \\
    \tilde{\eta}_0 &= -n\{\bar{X}(t)-\mu(t)\}(\bar{Y}-\E(Y)), \\
    \rho_0 &= \sum_{i=1}^{n}\left(\hat{X}_i^{(q)}-X_i(t)+\bar{X}(t) - \bar{\hat{X}}^{(q)}(t)\right)(Y_i-\bar{Y}).
\end{align*}
Note that $\rho_0=O_p\{n^{1/2}h^{-2}\}$, since
    \begin{align*}
        \rho_0 &= \sum_{i=1}^{n}\left(\hat{X}_i^{(q)}-X_i(t)+\bar{X}(t) - \bar{\hat{X}}^{(q)}(t)\right)(Y_i-\bar{Y}) \\
        &= O_p\left\{n\left(n^{-1/2}+\max_k\{|B|_k\}+\sum_{\ell=q+1}^{\infty}|\xi_{i\ell}|+n^{-1/2}h^{-2}\right)\right\} \\
        &= O_p\{n^{1/2}h^{-2}\}.
    \end{align*}
Moreover, $\|\tilde{\eta}_0\|=O_p(1), \|\tilde{\zeta}_0\|=O_p(1)$, thus $n^{-1/2}\|\tilde{\zeta}_0\| + n^{-1}(\|\tilde{\eta}_0\|+\|\rho_0\|) = O_p\{n^{-1/2}h^{-2}\}$, and therefore $\hat{K}(b)$ can be written as
\begin{align*}
    \hat{K}(b)=K(b)+n^{-1/2}\zeta_0,
\end{align*}
where $\|\zeta_0\|=O_p(h^{-2})$.
 
By means of algebraic computation, we can sequentially write
\begin{align*}
    \hat{K}^{j+1}(b)(t) &= \int \hat{K}^j(b)(s)\hat{K}(s,t)\;ds \\
    &= \int \{K^j(b)+n^{-1/2}\zeta_j+n^{-1}\eta_j\}(s)\{K+n^{-1/2}h^{-2}\zeta\}(s,t)\;ds \\
    &=K^{j+1}(b)(t) + n^{-1/2}\int\{h^{-2}K^j(b)(s)\zeta(s,t) + \zeta_j(s)K(s,t)\}\;ds \\
    &\qquad\qquad\qquad + n^{-1}\int\{\eta_j(s)K(s,t) + h^{-2}\zeta_j(s)\zeta(s,t)\; ds\} + n^{-3/2}\int h^{-2}\eta_j(s)\zeta(s,t)\;ds \\
    &= K^{j+1}(b)(t) + n^{-1/2}\zeta_{j+1} + n^{-1}\eta_{j+1},
\end{align*}
where 
\begin{align*}
    \zeta_{j+1} &= \int\{h^{-2}K^j(b)(s)\zeta(s,t) + \zeta_j(s)K(s,t)\}\;ds, \\
    \eta_{j+1}(t) &= \int\{\eta_j(s)K(s,t) + h^{-2}\zeta_j(s)\zeta(s,t)\; ds\} + n^{-1/2}\int h^{-2}\eta_j(s)\zeta(s,t)\;ds,
\end{align*}
$\eta_1=\int h^{-2}\zeta(s,t)\zeta_0(s)\;ds$, and $\eta_0=0$.

First note that, 
\begin{align*}
    \|\eta_{j+1}\|&\leq\|\eta_j\|\left(\|K\|+n^{-1/2}h^{-2}\|\zeta\|\right) + h^{-2}\|\zeta\|\|\zeta_j\| \\
    &=\|\eta_j\|R_1 + \|\zeta_j\|R_2,
\end{align*}
where, $R_1=\|K\|+n^{-1/2}h^{-2}\|\zeta\|$ and $R_2=h^{-2}\|\zeta\|$. Thus,
\begin{align*}
    \|\eta_{j+1}\|&\leq \|\eta_j\|R_1+\|\zeta_j\|R_2 \\
    &\leq\left(\|\eta_{j-1}\|R_1+\|\zeta_{j-1}\|R_2 \right)R_1 + \|\zeta_j\|R_2 \\
    &= \|\eta_{j-1}\|R_1^2 + \|\zeta_{j-1}\|R_1R_2 + \|\zeta_j\|R_2 \\
    &\leq\|\eta_1\|R_1^{j} + \sum_{k=1}^{j}\|\zeta_{j-k}\|R_1^k R_2,
\end{align*}
where $\eta_1(t)=\int h^{-2}\zeta(s,t)\zeta_0(s)\;ds$, implying $\|\eta_1\|=O_p(h^{-4})$.
Therefor we have, 
\begin{align*}
    \|\eta_{j+1}\|\leq h^{-2}\|\zeta\|\sum_{k=0}^{j}\|\zeta_k\|R_1^{j-k} + \|\eta_1\|R_1^j.
\end{align*}
In addition, if we let $\vartheta_j := \int K^j(b)(s)\zeta(s,t)\;ds$, then we can write 
\begin{align*}
    \zeta_{j+1} = K(\zeta_j)(t) + h^{-2}\vartheta_j(t) = K^{j}(\zeta_0) + \sum_{\ell=0}^{j-1}K^\ell(\vartheta_{j-\ell}).
\end{align*}
Since, $\|\vartheta_j\|\leq \|K^j(b)\|\|\zeta\|$, we can obtain the following upper bound for $\|\zeta_j\|$.
\begin{align*}
    \|\zeta_j\|\leq \|K^{j-1}\|\|\zeta_0\|+h^{-2}\|\zeta\|\sum_{\ell=1}^{j-2}\|K^{\ell}\|\|K^{j-\ell-1}(b)\|.
\end{align*}
By incorporating the upper bound obtained for $\|\zeta_j\|$ we arrive at the following upper bound for $\|\eta_{j+1}\|$.
\begin{align*}
    \|\eta_{j+1}\| \leq \|\eta_1\|R_1^j + h^{-2}\|\zeta\|\sum_{\ell=1}^{j}R_1^{j-k}\|K^{\ell-1}\|\zeta_0\| + h^{-4}\|\zeta\|^2\sum_{k=1}^{j}R_1^{j-k}\sum_{\ell=1}^{k-2}\|K^\ell\|\|K^{k-\ell-1}(b)\|.
\end{align*}
Note that, since $K(s,t) = \sum_{\ell=1}^{\infty}\theta_\ell\phi_\ell(s)\phi_\ell(t)$, $\|K^\ell\|^2 = O(\theta_1^{2\ell})$, $\|K^\ell(b)\| = O(\theta_1^{2\ell})$, and recall that  $\|\zeta\| = O_p(1)$ as $n\to\infty$. This implies that,  $\|\vartheta_j\|\leq O_p(\theta_1^j)$, $\|\zeta_j\|\leq O_p(h^{-2}j\|K\|^j)$, $R_1^j = \{\|K\|+n^{-1/2}h^{-2}\|\zeta\|\}^j= O_p(\|K\|^j)$, and 
\begin{align*}
    \|\eta_j\| &= O_p\left(\|\eta_1\|\|K\|^j + h^{-4}\sum_{k=1}^{j}\|K\|^{j-k-1}\|K\|^k + h^{-4}\sum_{k=1}^{j-1}\|K\|^{j-k-1}k\|K\|^k\right) \\
    &= O_p\left(h^{-4}\|K\|^j\right).
\end{align*}

This implies that,
\begin{align*}
    \|\hat{K}^j(b)-K^j(b)\|\leq n^{-1/2}\|\zeta_j\| + n^{-1}\|\eta_j\| &= O_p\left(n^{-1/2}(1+h^{-2}j)\theta_1^j + n^{-1}h^{-4}\|K\|^j\right) \\
    &= O_p\{n^{-1/2}h^{-2}\|K\|^{j}\}
\end{align*}
uniformly in $1\leq j \leq n^{1/2}$, given that $n^{-1/2}h^{-4}\to 0$.

\end{proof}

\subsection*{III. Convergence rate of matrix entries $\hat{h}_{jk}$ to $h_{jk}$}
\begin{proposition}
If $n^{-1/2}h^{-4}\to 0$ as $n\to\infty$ and $h\to 0$, then
    \begin{align*}
        \hat{h}_{jk} = h_{jk} &+ n^{-1/2}\int\{K^{k}(b)+K^{j+1}(b)\zeta_k\} \\
        &+ n^{-1}O_p\left\{h^{-4}jk\|K\|^{j+k} + n^{-1/2}h^{-6}k\|K\|^{j+k}\right\}
    \end{align*}
    uniformly in $1\leq j \leq k \leq n^{1/2}$.
\end{proposition}
\begin{proof}

Asymptotic expansion of $\hat{h}_{jk}$ around $h_{jk}$ follows along the same lines as in the proof of Theorem 5.2 in Delaigle and Hall (2012) as outlined below.

By the definition of $\hat{h}_{jk}$ and $h_{jk}$, 
\begin{align*}
    \hat{h}_{jk} &= \int \hat{K}^{j+1}(b)\hat{K}^k(b) \\
    &= \int \{K^{j+1}(b)+n^{-1/2}\zeta_{j+1} + n^{-1}\eta_{j+1}\}\{K^k(b)+n^{-1/2}\zeta_k + n^{-1}\eta_k\} \\
    &= h_{jk} + n^{-1/2}\int\{\zeta_{j+1}K^k(b) + K^{j+1}(b)\zeta_k\} + n^{-1}R_{jk},
\end{align*}
where 
\begin{align*}
    R_{jk} := \int\{\zeta_{j+1}\zeta_k+K^{j+1}(b)\eta_k+K^k(b)\eta_{j+1}\} + n^{-1/2}\int\{\eta_{j+1}\zeta_k+\zeta_{j+1}\eta_k\} + n^{-1}\int\eta_{j+1}\eta_k.
\end{align*}
Incorporating the bound obtained for $\|\eta_j\|, \|\zeta_j\|, \|K\|, \|K(b)\|, \|\zeta\|$, we can write
\begin{align*}
    R_{jk} &= O_p\left[h^{-4}jk\|K\|^{j+k} + h^{-4}\theta_1^j\|K\|^k+h^{-4}\theta_1^k\|K\|^j + \right. \\
    &\left.\qquad\qquad + n^{-1/2}\left\{\|K\|^jh^{-4}h^{-2}k\|K\|^k+ \|K\|^kh^{-4}h^{-2}j\|K\|^j\right\} + n^{-1}h^{-8}(\|K\|^{j+k})\right] \\
    &= O_p\left\{h^{-4}jk\|K\|^{j+k}+n^{-1/2}h^{-6}k\|K\|^k\|K\|^j+n^{-1}h^{-8}\|K\|^{j+k}\right\} \\
    &= O_p\left\{h^{-4}jk\|K\|^{j+k} + n^{-1/2}h^{-6}k\|K\|^k\|K\|^j\right\}
\end{align*}
uniformly in $1\leq j \leq k \leq n^{1/2}$. 

\end{proof}

\subsection*{IV. Convergence rate of $\hat{\gamma}$ to $\gamma$}
\begin{proposition}
Let $\lambda=\lambda_p(H)$ be the smallest eigenvalue of $H$. Then,
    \begin{align*}
        \|\hat{\gamma}-\gamma\|=O_p\left\{\max\{n^{-1/2}h^{-4}, n^{-1}h^{-6}\}\lambda^{-1}\|\gamma\| + \max\{n^{-1}h^{-8},n^{-2}h^{-12}\}\lambda^{-3}\right\}.
    \end{align*}
\end{proposition}

\begin{proof}
Let $\|M\|$ for a matrix $M$ represents its operator norm. 
Note that, $\|\Delta H^{-1}\|\leq \|\Delta\|/\lambda$. Let, $\hat{H} = [\hat{h}_{jk}], \Delta=\hat{H}-H$. Provided that $\|\Delta\|/\lambda\leq\rho$, for a fixed $\rho\in(0,1)$,
\begin{align*}
    \hat{H}^{-1} = (I+H^{-1}\Delta)^{-1}H^{-1} = O_p\left[I-H^{-1}\Delta+O_p\{(\|\Delta\|/\lambda)^2\}\right]H^{-1}.
\end{align*}
Note that, $\hat{h}_{jk} = h_{jk} + n^{-1/2}\Delta_{1jk} + n^{-1}\Delta_{2jk}$ where 
\begin{align*}
    |\Delta_{1jk}|&=\left|\int\zeta_{j+1}K^k(b) + K^{j+1}(b)\zeta_k\right| \\
    &\leq \|\zeta_{j+1}\|\|K^k(b)\| + \|K^{j+1}(b)\|\|\zeta_k\| \\
    &\leq O_p\left[\max\{j,k\}h^{-2}\|K\|^{j+k}\right],
\end{align*}
and 
\begin{align*}
    \|\Delta_{2jk}\|= O_p\left\{h^{-4}jk\|K\|^{j+k} + n^{-1/2}h^{-6}k\|K\|^{j+k}\right\}.
\end{align*}
This implies that, when $p$ diverges slowly that $p=O_p(n^{1/2})$, then 
\begin{align*}
    \sum_{j=1}^{p}\sum_{k=1}^{p}\Delta_{1jk}^2 &= O_p(h^{-4}), \\
    \sum_{j=1}^{p}\sum_{k=1}^{p}\Delta_{2jk}^2 &= O_p\left[\max\{h^{-8}, n^{-1}h^{-12}\}\right], \\
    n\sum_{j=1}^{p}\sum_{k=1}^{p}\Delta_{jk}^2 &= O_p\left[\max\{h^{-8}, n^{-1}h^{-12}\}\right].
\end{align*}
Thus, $\|\Delta\| = O_p\left\{\max\{n^{-1/2}h^{-4},n^{-1}h^{-6}\}\right\}$, implying that
\begin{align*}
    \hat{H}^{-1} = \left\{I-H^{-1}\Delta+O_p\left\{\max\{n^{-1}h^{-8},n^{-2}h^{-12}\}\right\}\lambda^{-2}\right\}H^{-1}.
\end{align*}
Also, 
\begin{align*}
    \|\hat{\alpha}\| &\leq \left\{\sum_{j=1}^{p}\|\hat{K}(b)\|^2\|\hat{K}^j(b)\|^2\right\}^{1/2} \\
    &= O_p\left[\left\{\sum_{j=1}^{p} \left\{\theta_1^j+n^{-1/2}h^{-2}j\|K\|^j+n^{-1}h^{-4}\|K\|^j\right\}^2 \left\{\theta_1+n^{-1/2}h^{-2}\zeta_0\right\}^2 \right\}^{1/2}\right] \\
    &= O_p(1)
\end{align*}
Next, recall that $\alpha_j = h_{0j}$, therefore by our previous analysis, $\hat{\alpha}_j=\alpha_j+n^{-1/2}\Delta_{10j}+n^{-1}\Delta_{20j}$, where $\sum_{j=1}^{p}|\Delta_{10j}|^2=O_p\{h^{-4}\}$ and $\sum_{j=1}^{p}|\Delta_{20j}|^2=O_p\{h^{-8}+n^{-1}h^{-12}+n^{-1/2}h^{-10}\}$, given $n^{-1/2}h^{-2}\to 0$, we have $\hat{\alpha} = \alpha + n^{-1/2}h^{-2}\delta$ where $\|\delta\| = O_p(1)$.
Moreover, recall that $\hat{\gamma} = \hat{H}^{-1}\hat{\alpha}$, where $\hat{\alpha}=[\hat{\alpha}_1,\dots,\hat{\alpha}_p]^\prime$. We than have  
\begin{align*}
    \hat{\gamma} &=  \left[I-H^{-1}\Delta+O_p\left\{\max\{n^{-1}h^{-8},n^{-2}h^{-12}\}\right\}\lambda^{-2}\right]H^{-1}\left(\alpha+n^{-1/2}h^{-2}\delta\right) \\
    &=  \left[H^{-1}-H^{-1}\Delta H^{-1}+O_p\left\{\max\{n^{-1}h^{-8},n^{-2}h^{-12}\}\right\}\lambda^{-2}H^{-1}\right]\left(\alpha+n^{-1/2}h^{-2}\delta\right) \\
    &= H^{-1}\alpha + H^{-1}(n^{-1/2}h^{-2}\delta)-H^{-1}\Delta H^{-1}\alpha - H^{-1}\Delta H^{-1}(n^{-1/2}h^{-2}\delta) \\
    &\qquad\qquad + O_p\left\{\max\{n^{-1}h^{-8},n^{-2}h^{-12}\}\lambda^{-3}+ \max\{n^{-1}h^{-8},n^{-2}h^{-12}\}\lambda^{-3}n^{-1/2}h^{-2}\right\} \\
    &= \gamma+H^{-1}\left(n^{-1/2}h^{-2}\delta-\Delta\gamma\right) + O_p\left\{\max\{n^{-1}h^{-8},n^{-2}h^{-12}\}\lambda^{-3}\right\}.  
\end{align*}
This implies that
\begin{align*}
        \|\hat{\gamma}-\gamma\|=O_p\left\{\max\{n^{-1/2}h^{-4}, n^{-1}h^{-6}\}\lambda^{-1}\|\gamma\| + \max\{n^{-1}h^{-8},n^{-2}h^{-12}\}\lambda^{-3}\right\}.
    \end{align*}
\end{proof}

\subsection*{V. Stochastic expansion of $\hat{g}_p(x)-g_p(x)$}
Below we present the stochastic expansion of the difference between $g_p(x)$ and its estimator $\hat{g}_p(x)$ . 

\begin{proposition} \label{prop:estimation_rate}
If $n^{-1/2}h^{-4}\to 0$ as $n\to\infty$ and $h\to 0$, then
    \begin{equation}
        \begin{split}
            &\hat{g}_p(x)-g_p(x)-(\bar{Y}-\E[Y]) = \\
            &\sum_{j=1}^{p}\left[(\hat{\gamma}_j-\gamma_j)\int (x-\E(X))K^j(b) \;+ \right. \\
    &\left. \qquad\qquad \gamma_j\int \left\{n^{-1/2}(x-\E(X))\zeta_j - (\bar{X}-\E(X))K^j(b)\right\} \right] \;+ \\
            &O_p\left[\max\{n^{-1}h^{-8},n^{-2}h^{-12}\}\lambda^{-3}+n^{-1}h^{-4}\max\{n^{-1/2}h^{-4},n^{-1}h^{-6}\}\lambda^{-1}\|\gamma\|\right] . 
        \end{split}
    \end{equation}
    This implies that 
    \begin{align*}
        \|\hat{g}_p(x)-g_p(x)\| =O_p\left\{n^{-1/2}h^{-4}\lambda^{-1}+n^{-1}h^{-8}\lambda^{-3}\right\}.
    \end{align*}
\end{proposition}

\begin{proof}
    
Note that, 
\begin{align*}
    \hat{g}_p(x)-g_p(x)&-(\bar{Y}-\E(Y)) \\
    &= \sum_{j=1}^{p}\left[(\hat{\gamma}_j-\gamma_j)\int (x-\E(X))K^j(b) \;+ \right. \\
    &\left. \qquad\qquad \gamma_j\int (x-\E(X))\{\hat{K}^j(b)-K^j(b)\} - \gamma_j\int (\bar{X}-\E(X))K^j(b) \right] \;+ \\
    & \quad\sum_{j=1}^{p}\left[(\hat{\gamma}_j-\gamma)\int(x-\E(X))\{\hat{K}^j(b)-k^j(b)\} \; - \right.\\
    &\left. \qquad\qquad (\hat{\gamma}_j-\gamma)\int(\bar{X}-\E(X))K^j(b)-\hat{\gamma}_j\int(\bar{X}-\E(X))\{\hat{K}^j(b)-K^j(b)\}\right].
\end{align*}
Note that, the second summation is of order
\begin{align*}
    O_p\left\{n^{-1/2}\sum_{j=1}^{p}\left(|\hat{\gamma}_j-\gamma_j|+n^{-1/2}|\hat{\gamma}_j|\right)\left(n^{-1/2}h^{-2}\|K\|^j\right)\right\}.
\end{align*}
With this, and recalling the asymptotic expansion $\hat{K}^j(b)=K^{j}(b)+n^{-1/2}\zeta_j+n^{-1}\eta_j$ and noting that $\|\eta\|_j=O_p\{h^{-4}\|K\|^{j}\}$, we can write $\hat{g}_p(x)-g_p(x)-(\bar{Y}-\E(Y))$ as
\begin{align*}
    \hat{g}_p(x)-g_p(x)&-(\bar{Y}-\E(Y)) \\
    &= \sum_{j=1}^{p}\left[H^{-1}(n^{-1}h^{-2}\delta-\Delta\gamma)\int (x-\E(X))K^j(b) \;+ \right. \\
    &\left. \qquad\qquad \gamma_j\int \left\{n^{-1/2}(x-\E(X))\zeta_j - (\bar{X}-\E(X))K^j(b)\right\} \right] \;+ \\ 
    &\qquad O_p\left[\max\{n^{-1}h^{-8},n^{-2}h^{-12}\}\lambda^{-3}\int(x-\E(X))K^j(b)\right] \;+ \\
    &\qquad O_p\left[n^{-1}\gamma_j\sum_{j=1}^{p}\int(x-\E(X))h^{-4}\|K\|^j\right]\;+ \\
    &\qquad O_p\left[n^{-1/2}\sum_{j=1}^{p}\left(|\hat{\gamma}_j-\gamma_j|+n^{-1/2}|\hat{\gamma}_j|\right)\left(n^{-1/2}h^{-2}\|K\|^j\right)\right].
\end{align*}
Let 
\begin{align*}
    (*) &= \sum_{j=1}^{p}\left[(\hat{\gamma}_j-\gamma_j)\int (x-\E(X))K^j(b) +  \gamma_j\int \left\{n^{-1/2}(x-\E(X))\zeta_j - (\bar{X}-\E(X))K^j(b)\right\} \right] \\
    (**) &= O_p\left[\max\{n^{-1}h^{-8},n^{-2}h^{-12}\}\lambda^{-3}\int(x-\E(X))K^j(b)\right] \\
    (***) &= O_p\left[n^{-1}\gamma_j\sum_{j=1}^{p}\int(x-\E(X))h^{-4}\|K\|^j\right]\;+ \\
    &\;\quad O_p\left[n^{-1/2}\sum_{j=1}^{p}\left(|\hat{\gamma}_j-\gamma_j|+n^{-1/2}|\hat{\gamma}_j|\right)\left(n^{-1/2}h^{-2}\|K\|^j\right)\right].
\end{align*}
By incorporating the asymptotic bound for $\|\hat{\gamma}-\gamma\|$ and  $\|\Delta\|$, we have
\begin{align*}
    (*) &= O_p\left\{\lambda^{-1}n^{-1/2}h^{-2} + \lambda^{-1}\max\{n^{-1/2}h^{-4},n^{-1}h^{-6}\}\right\} \\
    &= O_p\left[\max\{n^{-1/2}h^{-4},n^{-1}h^{-6}\}\lambda^{-1}\right],
\end{align*}
\begin{align*}
    (***) = O_p\left\{n^{-1}h^{-4}\|K\|^j\|\hat{\gamma}-\gamma\|+n^{-1}h^{-4}\|\gamma\|\right\},
\end{align*}

\begin{align*}
    (**)+(***) &= O_p\left[\max\{n^{-1}h^{-8},n^{-2}h^{-12}\}\lambda^{-3}\right] \;+ \\    &O_p\left[n^{-1}h^{-4}h^{-2}n^{-1/2}\lambda^{-1}+n^{-1}h^{-4}\lambda^{-1}\|\gamma\|\max\{n^{-1/2}h^{-4},n^{-1}h^{-6}\} \;+ \right.\\
    &\qquad\left. n^{-1}h^{-4}\max\{n^{-1}h^{-8},n^{-2}h^{-12}\lambda^{-3}\}\right] \\
    &= O_p\left[\max\{n^{-1}h^{-8},n^{-2}h^{-12}\}\lambda^{-3}+n^{-1}h^{-4}\max\{n^{-1/2}h^{-4},n^{-1}h^{-6}\}\lambda^{-1}\|\gamma\|\right].
\end{align*}
Thus, given $n^{-1/2}h^{-4}\to 0$,
\begin{align*}
    (*)+(**)+(***) &= O_p\left[\max\{n^{-1/2}h^{-4},n^{-1}h^{-6}\}\lambda^{-1}+n^{-1}h^{-4}\max\{h^{-4},n^{-1}h^{-8}\}\lambda^{-3}\right] \\
    &= O_p\left\{n^{-1/2}h^{-4}\lambda^{-1}+n^{-1}h^{-8}\lambda^{-3}\right\}.
\end{align*}

\end{proof}

\begin{proof}[Proof of Theorem \ref{thm:predictive_consistency}]
    Proof follows along the same line as in the proof of (5.11) in \cite{Delaigle_Hall_2012}, using the rate obtained in Proposition \ref{prop:estimation_rate}.
\end{proof}

\section{Additional Simulations}\label{appendix:extra_simulation}
In this section we present a numerical investigation of the effect of the smoothing parameter $h$ on the predictive performance of $P^3LS$ for the simulation setting presented in Section 3.1 of the main manuscript. To do so, we applied 4-fold cross validation and computed the average rMSPE of the 4 folds. More precisely, for a given $h$, and for the 4 index sets $\mathcal{M}_1=\{1,\dots,100\}-\{1,\dots,25\}, \mathcal{M}_2=\{1,\dots,100\}-\{26,\dots,50\}, \mathcal{M}_3=\{1,\dots,100\}-\{51,\dots,75\}, \mathcal{M}_1=\{1,\dots,100\}-\{76,\dots,100\} $ we applied the $P^3LS$ to fit the predictive model (1) and used the estimated model to predict the remaining observations not included in the index set and computed the $rMSPE$s $\mathcal{E}_1,\dots,\mathcal{E}_4$, respectively. Finally, we computed the average $rMSPE$s of the 4 folds, i.e. $(\mathcal{E}_1+\mathcal{E}_1+\mathcal{E}_2+\mathcal{E}_3+\mathcal{E}_4)/4$. Box plots of the average $rMSPE$s described above as a function of $h$ are illustrated in Figures \ref{fig:Boxplots_X3Bc_fun2_h_cv}, \ref{fig:Boxplots_X3Bcon_fun2_h_cv}, \ref{fig:Boxplots_X3Bcon2_fun2_h_cv}, and \ref{fig:Boxplots_X3Bcon3_fun2_h_cv}. As plots illustrate, the average $rMSPE$s do not exhibit a notable change across the range of values considered for $h$. However, smaller values of $h$ seem to be favorable for prediction purposes.

\begin{figure}[H]
    \centering
    \includegraphics[width=6.5in, height=3.5in]{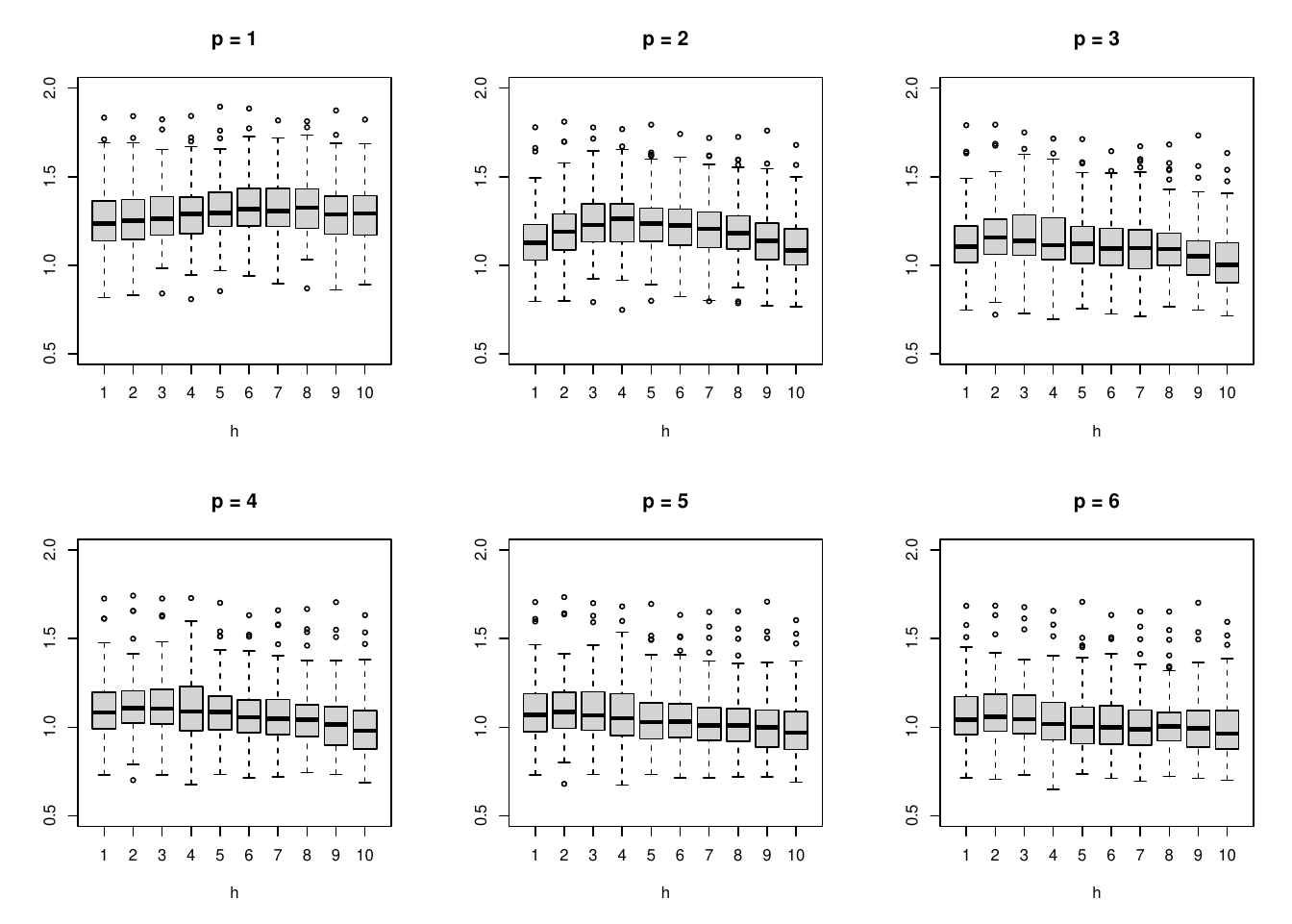}
    \caption{Case 1}
    \label{fig:Boxplots_X3Bc_fun2_h_cv}
\end{figure}

\begin{figure}[H]
    \centering
    \includegraphics[width=6.5in, height=3.5in]{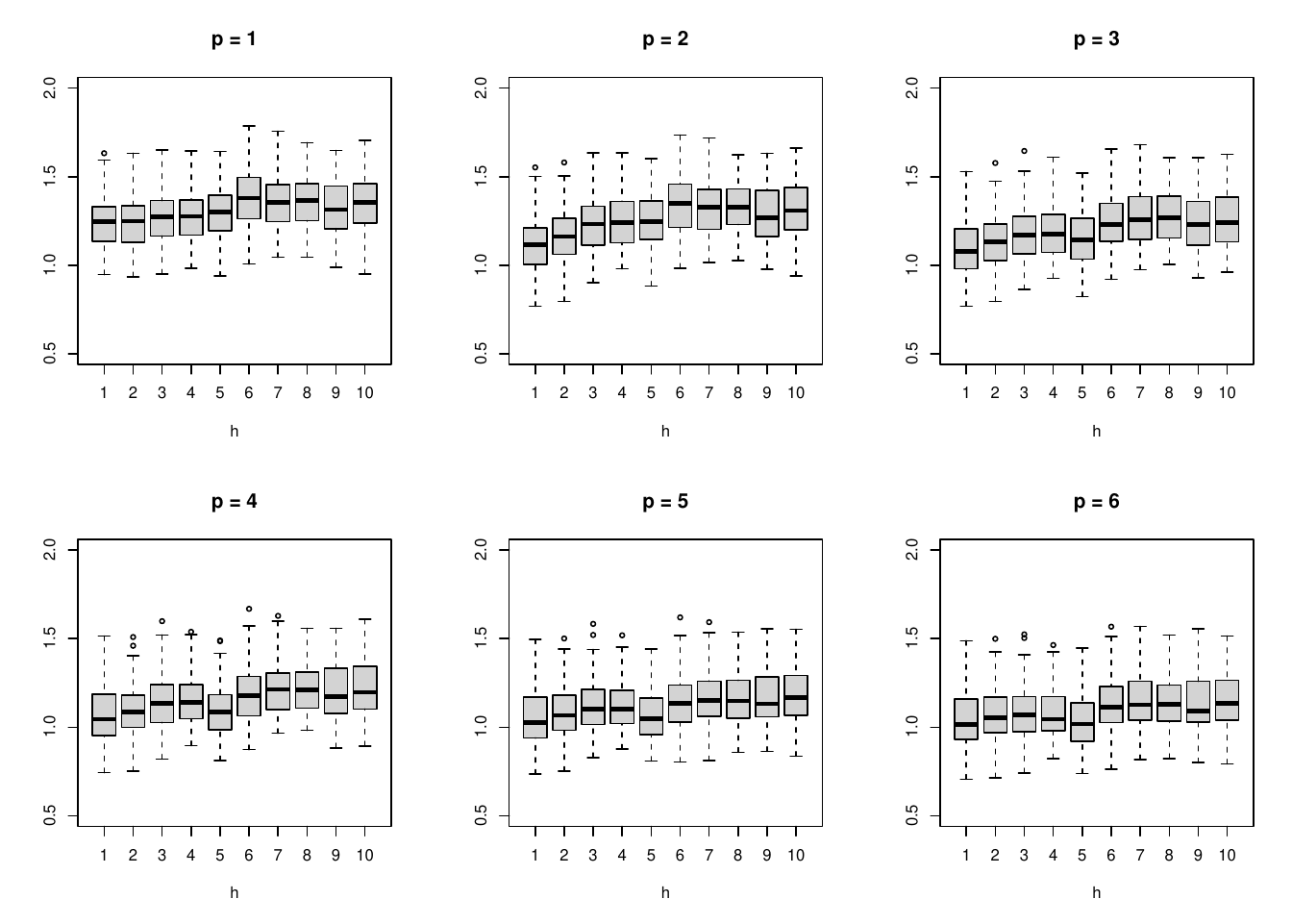}
    \caption{Case 2}
    \label{fig:Boxplots_X3Bcon_fun2_h_cv}
\end{figure}

\begin{figure}[H]
    \centering
    \includegraphics[width=6.5in, height=3.5in]{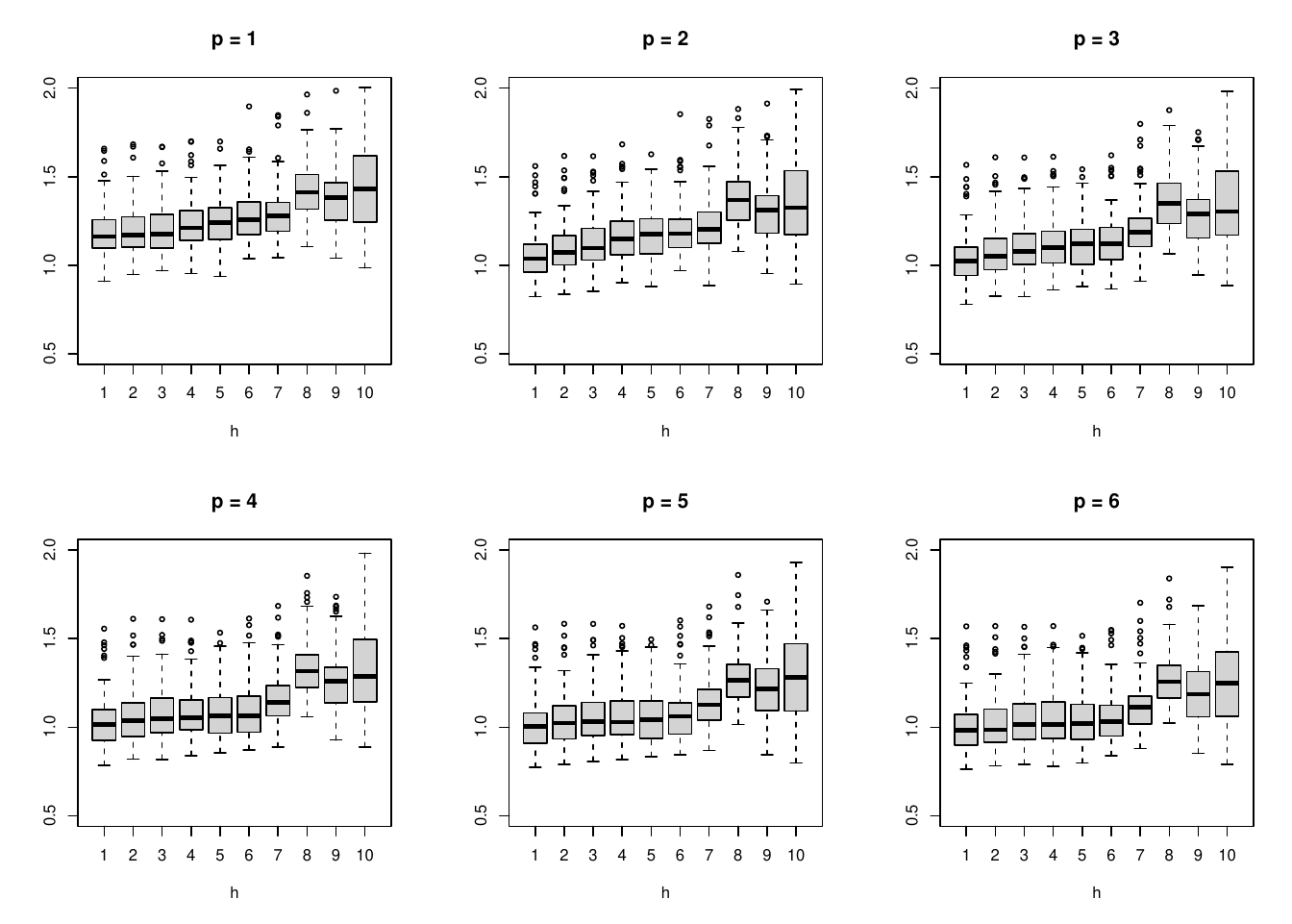}
    \caption{Case 3}
    \label{fig:Boxplots_X3Bcon2_fun2_h_cv}
\end{figure}

\begin{figure}[H]
    \centering
    \includegraphics[width=6.5in, height=3.5in]{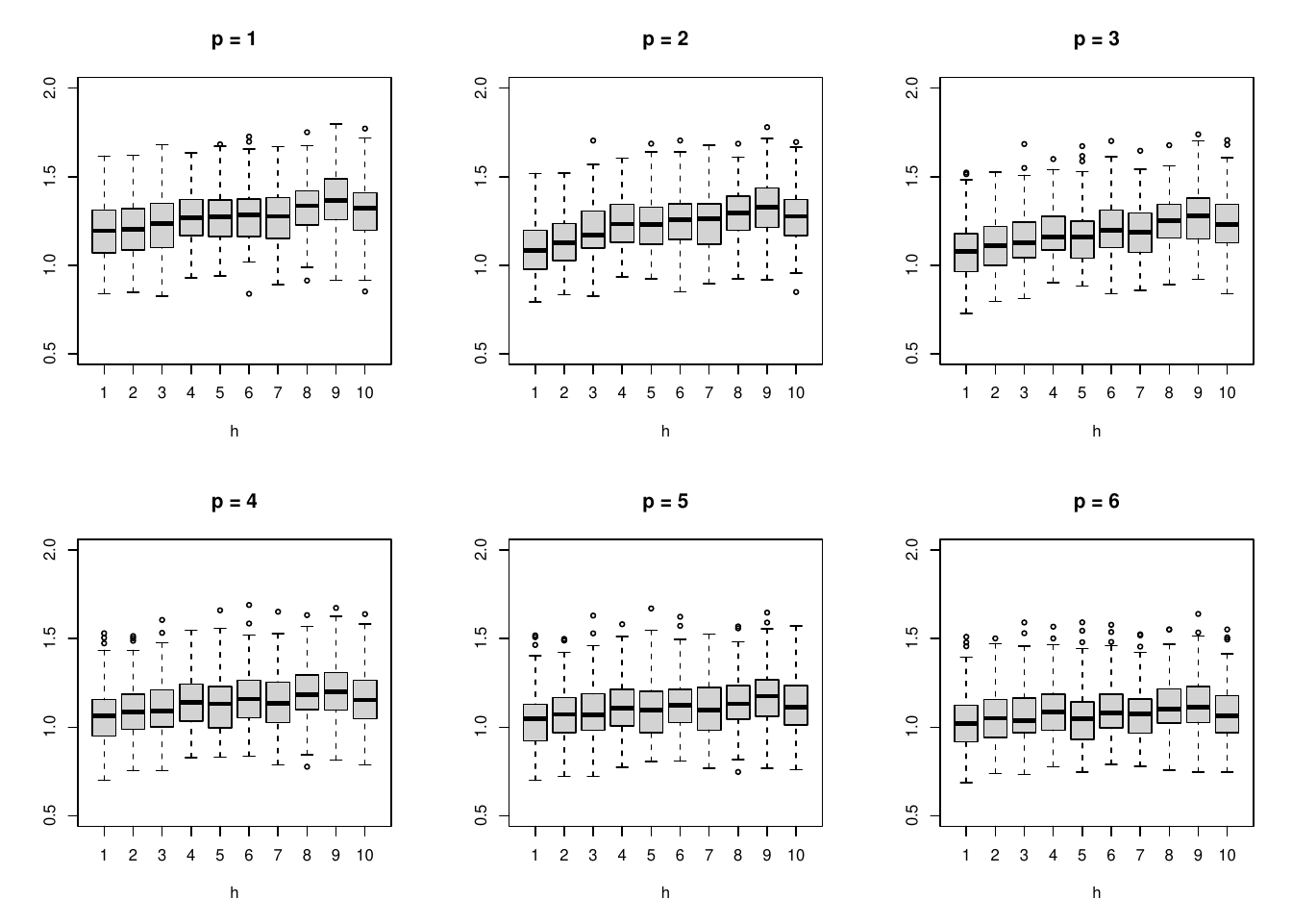}
    \caption{Case 4}
    \label{fig:Boxplots_X3Bcon3_fun2_h_cv}
\end{figure}

\label{lastpage}

\end{appendices}

\end{document}